\documentclass[12pt]{article}

\usepackage{hyperref}
\usepackage[font={small,it}]{caption}
\usepackage[ruled,vlined]{algorithm2e}
\usepackage{amssymb}
\usepackage{amsfonts}
\usepackage{multirow}
\usepackage{amsthm}
\usepackage{xcolor}

\newtheorem{theorem}{Theorem}
\newtheorem{lemma}{Lemma}

\newtheorem{proposition}{Proposition}
\theoremstyle{definition}

\newtheorem{remark}{Remark}
\newtheorem{assumption}{Assumption}
\def\code#1{\texttt{#1}} 
\def\T{{ \mathrm{\scriptscriptstyle T} }} 

\usepackage{amsmath}
\usepackage{graphicx}
\usepackage{enumerate}
\usepackage{natbib}
\usepackage{url} 
\newcommand{\blind}{0}

\begin{document}

\def\spacingset#1{\renewcommand{\baselinestretch}%
{#1}\small\normalsize} \spacingset{1}


\if0\blind
{
  \title{\bf Confidence Intervals for Parameters in High-dimensional Sparse Vector Autoregression}
    \author{
    Ke Zhu and Hanzhong Liu
    \\ \\
    Center for Statistical Science, Department of Industrial Engineering, \\
    Tsinghua University, Beijing, China
    }
  \maketitle
} \fi

\if1\blind
{
  \bigskip
  \bigskip
  \bigskip
  \begin{center}
    {\LARGE\bf Confidence Intervals for Parameters in High-dimensional Sparse Vector Autoregression}
\end{center}
  \medskip
} \fi

\bigskip
\begin{abstract}
Vector autoregression (VAR) models are widely used to analyze the interrelationship between multiple variables over time. Estimation and inference for the transition matrices of VAR models are crucial for practitioners to make decisions in fields such as economics and finance. However, when the number of variables is larger than the sample size, it remains a challenge to perform statistical inference of the model parameters. In this article, we propose the de-biased Lasso and two bootstrap de-biased Lasso methods to construct confidence intervals for the elements of the transition matrices of high-dimensional VAR models. We show that the proposed methods are asymptotically valid under appropriate sparsity and other regularity conditions. To implement our methods, we develop feasible and parallelizable algorithms, which save a large amount of computation required by the nodewise Lasso and bootstrap. A simulation study illustrates that our methods perform well in finite samples. Finally, we apply our methods to analyze the price data of stocks in the S\&P 500 index in 2019. We find that some stocks, such as the largest producer of gold in the world, Newmont Corporation, have significant predictive power over the most stocks.
\end{abstract}

\noindent
{\it Keywords:}  Bootstrap, De-biased Lasso, De-sparsified Lasso, Granger Causality,  High-dimensional time series
\vfill

\newpage
\spacingset{1.45} 

\section{Introduction}
\label{sec:intro}

Vector autoregression (VAR) models have been widely used in econometric, business statistics, and other fields \citep{Sims:1980tu,Fuller:1996wa,Lutkepohl:2007vj}. These models can capture the dynamic relationship between factors through transition matrices. In practice, the number of factors can be large, even larger than the sample size. For instance, the number of stocks in the market is usually larger than the number of observations. This fact leads researchers to consider high-dimensional VAR models.

High-dimensional VAR models have been thoroughly studied in the past decades. For instance, \cite{Guo:2016hi} imposed banded structure on the transition matrices and established the convergence rates of the least squares estimators. A series of works in the literature has proposed the sparsity constraint on the transition matrices and used the Lasso \citep{Tibshirani:1996} to estimate the parameters \citep{Hsu:2008,Song:2011,Negahban:2011dm,Loh:2012if,Chen:2013fu,Han:2015vc,Kock:2015bb,Basu:2015ho,Davis:2016fq}. Under  regularity conditions, 
\cite{Basu:2015ho} established the deviation bounds of the Lasso estimators, which are essential for studying the theoretical properties of Lasso-based estimators in high-dimensional sparse VAR models.

In addition to parameter estimation, confidence intervals and hypothesis testing provide practitioners, such as, policy makers and business owners, with more valuable and solid information for assessing the significance of the correlations between factors.
For instance, \cite{wilms2016predictive} used  Granger causality test \citep{granger1969investigating} based on the bootstrap adaptive Lasso to detect the most predictive industry-specific economic sentiment indicators for macro-economic indicators. 
\cite{lin2017regularized} proposed testing procedures in multi-block VAR models to test whether a block ``Granger-causes'' another block of variables and applied them to analyze the temporal dynamics of the S\&P100 component stocks and key macroeconomic factors.
In the context of VAR models, the Granger causality test is equivalent to test if the element of transition matrices is equal to zero \citep{Lutkepohl:2007vj}. There is an increasing demand for constructing confidence intervals or performing hypothesis testing for the element of transition matrices in high-dimensional VAR models. Since the limiting distribution of the Lasso is complicated \citep{Knight:2000}, we cannot directly use it for statistical inference. In high-dimensional sparse linear regression models, a series of studies has proposed methods based on the Lasso or its variants for statistical inference. One direction of research has proposed the de-biased Lasso method, which focuses on obtaining an asymptotically normal estimator by correcting the bias of the Lasso estimator \citep{Zhang:2014,vandeGeer:2014,Javanmard:2014}. We will refer the method in \cite{Zhang:2014} and \cite{vandeGeer:2014} as LDPE (low dimensional projection estimator) and the method in \cite{Javanmard:2014} as JM. These estimators are also called the de-sparsified Lasso since they are no longer sparse after correcting for the bias. Another direction of research uses the bootstrap \citep{Efron:1979ha} for inference, including the bootstrap threshold Lasso \citep{Chatterjee:2011}, bootstrap adaptive Lasso \citep{Chatterjee:2013}, bootstrap Lasso + ordinary least squares (OLS) \citep{Liu:2013} and bootstrap de-biased Lasso \citep{Dezeure:2017jla}. \cite{Ning:2017ee} generalized the de-biased Lasso method from linear regression to penalized M-estimators using the so called de-correlated score function. \cite{Neykov:2016to} generalized this method to estimating equation problems that are likelihood-free using the projection approach. Based on the framework of \cite{Ning:2017ee}, \cite{Zheng:2018} studied hypothesis testing for sub-Gaussian VAR models using the de-correlated method. \cite{Basu:2017vr} proposed using the de-biased Lasso estimator to perform inference for high-dimensional sparse VAR models but did not provide rigorous theoretical guarantees. Since the predictors in VAR models are random and exhibit complex correlation structure, it is uncertain whether the de-biased Lasso method is valid and it is challenging to provide theoretical guarantees. \cite{Krampe:2018wf} proposed a model-based bootstrap de-biased Lasso estimator to perform inference, but its computational burden is heavy when the dimension $p$ is large. Moreover, in order to guarantee the required order of sparsity in the bootstrap sample, the study used the thresholded Lasso which is complicated than the Lasso although it is easier for theoretical study.

To fill in the theoretical gap of the de-biased Lasso estimator and to address the computation problem of the model-based bootstrap de-biased Lasso estimator, in this study, we formally establish the asymptotic normality of the de-biased Lasso estimator for inferring the elements of the transition matrix of high-dimensional sparse VAR(1) models. Then, we propose two computational feasible bootstrap methods, residual bootstrap de-biased Lasso (BtLDPE) and multiplier wild bootstrap de-biased Lasso (MultiBtLDPE). Our methods can be generalized to VAR($k$) models through  transformation. Our contributions are summarized as follows.

  First, we derive the asymptotic properties of the de-biased Lasso, residual bootstrap de-biased Lasso and multiplier wild bootstrap de-biased Lasso estimators in high-dimensional sparse VAR models by using the deviation bounds of the Lasso estimator \citep{Basu:2015ho} and the martingale central limit theorem (Theorem 5.3.4 in \cite{Fuller:1996wa}). We demonstrate the validity and robustness of these methods for statistical inference in a context broader than linear regression models. As a by-product, we show that the sparse Riesz condition holds for high-dimensional sparse VAR models and provide an upper bound on the sparsity of the Lasso estimator which is essential for obtaining appropriate variance estimators and studying the theoretical properties of bootstrap de-biased Lasso.

  Second, to implement our methods, we provide algorithms which are feasible and easy to parallel. Specifically, for each of the $p$ equations of the VAR models, we perform statistical inference separately. Since the $p$ equations share the same design matrix, we need only run nodewise lasso once, which is the main computational burden of de-biased Lasso and bootstrap de-biased Lasso. Compared to the model-based bootstrap de-biased Lasso, our methods have significant computational advantages especially when $p$ is large.
  
  Third, we conduct comprehensive simulation studies to compare our methods with the bootstrap Lasso, bootstrap Lasso+OLS and another de-biased Lasso method proposed by \cite{Javanmard:2014}. We find that the de-biased Lasso method, the LDPE, can always yield honest coverage probabilities, and when the sample size is large, two bootstrap methods, the BtLDPE and MultiBtLDPE, can yield the honest coverage probabilities with shorter interval lengths. We also apply our methods to analyze the S\&P 500 constituent stocks data set in 2019. We find that the prices of some stocks have significant predictive power over the prices of most stocks, for example, Newmont Corporation, which is the largest producer of gold in the world and the only gold producer listed in the S\&P 500 Index, has the ability to affect many other stock prices in advance.

\textbf{Notation}. 
For a vector $a=(a_1,...,a_n)$, we denote $\|a\|_{1}=\sum_{i=1}^{n}\left|a_{i}\right|$, $\|a\|_{2}^{2}=\sum_{i=1}^{n} a_{i}^{2}$ and $\|a\|_{\infty}= \max_i \left|a_{i}\right|$. We use $e_j$ to denote the vector whose $j$th element is one, zero otherwise.
For a square matrix $A$, let $\Lambda_{\rm min}(A)$ and $\Lambda_{\rm max}(A)$ be the smallest and largest eigenvalues of $A$ respectively. Let $|A|$ denote the determinant of $A$. 
For a design matrix $\mathbf{X}$, let $X_{tj}$ denote the $(t,j)$th element of $\mathbf{X}$, $X_j$ denote the $j$th column of $\mathbf{X}$ and $\mathbf{X}_{-j}$ denote the $\mathbf{X}$ without the $j$th column. Let $||\mathbf{X}||_{\infty}$ denote the maximum absolute value of the elements of $\mathbf{X}$.
For a set $S$, let $|S|$ denote the number of elements of $S$ and $\mathbf{X}_S$ denote the selected columns of $\mathbf{X}$ in $S$.
We use i.i.d as the abbreviation of independent and identically distributed. Let $z_\alpha$ denotes the lower $\alpha$ quantile of the standard normal distribution. 
For two sequences of positive numbers $\{a_n\}$ and $\{b_n\}$, we denote $a_n=o(b_n)$ if $a_{n} / b_{n} \rightarrow 0$ as $n \rightarrow \infty$ and $a_n=O(b_n)$  if $\limsup \left|a_{n} / b_{n}\right|<\infty$. We denote $a_{n} \asymp b_{n}$ if there exist positive constants $C_1,C_2>0$ such that $\limsup \left|a_{n} / b_{n}\right|\leq C_1$ and $\liminf \left|a_{n} / b_{n}\right|\geq C_2$. We denote $Z_n=o_p(1)$ if $Z_n$ convergent to zero in probability and $Z_n=O_p(1)$ if $Z_n$ is bounded in probability. 

The rest of the paper is organized as follows. 
In Section 2, we introduce the high-dimensional sparse VAR models and the de-biased Lasso, residual bootstrap de-biased Lasso and multiplier wild bootstrap de-biased Lasso methods.
In Section 3, we provide theoretical results, including the upper bound of the number of selected variables by the Lasso estimator and the asymptotic normality of the proposed estimators. 
In Section 4, simulation studies are provided for investigating the finite-sample performance of different methods. 
In Section 5, we illustrate our method using the S\&P 500 constituent stocks data set.
We summarize the results and discuss possible extensions in the last section.
Proof details are given in the supplementary material.

\section{Methods}
\label{sec:method}

In this section, we first introduce the high-dimensional sparse VAR models. Consider a $p$-dimensional VAR(1) model,
\begin{equation}\label{eq:VAR}
\mathbf{y}_t = A\mathbf{y}_{t-1} + \mathbf{e}_t,\quad t=1,...,n,
\end{equation}
where $\mathbf{y}_t=({y}_{1t},...,{y}_{pt})^{\T}$ is a $p$-dimensional random vector, 
$A=(a_{ij})_{p \times p}$ is a $p \times p$ transition matrix and $\mathbf{e}_t=({e}_{1t},...,{e}_{pt})^{\T}$ is a $p$-dimensional Gaussian white noise, namely, $\mathbf{e}_t$ is independently and identically distributed in multivariate Gaussian distribution $\mathcal{N}_p (\mathbf{0},\Sigma_\mathbf{e})$. 
We assume that the VAR process $\{\mathbf{y}_t\}$ is stable, namely, $\Lambda_{\rm max}(A)<1$. We have $n$ observations of $\mathbf{y}_t$ and $n$ may be smaller than $p$.

There are $p$ equations, corresponding to $p$ components of $\mathbf{y}_t$, in the VAR(1) model.
We consider each of the $p$ equations separately,
\begin{equation*}
y_{it} = a_i^{\T}\mathbf{y}_{t-1} + {e}_{it},\quad i=1,...,p,\quad t=1,...,n,
\end{equation*}
where $a_i=(a_{i1},...,a_{ip})$ is the $i$th row of $A$ and we denote the support set of $a_i$ by $S_i=\{j\in \{1,...,p\}:a_{ij}\ne0\}$ and let $s_{i}=|S_i|$. To gather all observations, we denote
\begin{equation}
\label{eqn:var1}
Y_{i} = \mathbf{X} a_i + {\varepsilon}_{i},\quad i=1,...,p,
\end{equation}
where $Y_{i}=({y}_{i1},...,{y}_{in})^{\T}$, $\mathbf{X}=(\mathbf{y}_0,...,\mathbf{y}_{n-1})^{\T}$ and $\varepsilon_i=(e_{i1},...,e_{in})^{\T}$. The denotation is to distinguish them from $\mathbf{y}_t$ and $\mathbf{e}_t$. Although equation \eqref{eqn:var1} violates the basic assumptions of linear regression models, it has the same form. Thus, we can apply the de-biased Lasso and bootstrap de-biased Lasso original proposed for linear regression models to construct confidence intervals for $a_i$'s, with caution that these methods may not be valid.

\subsection{De-biased Lasso}
\label{sec:dbl}

The Lasso is widely used for simultaneous parameter estimation and model selection in high-dimensional linear regression models \citep{Tibshirani:1996}, which adds an $l_1$ penalty to the loss function to obtain sparse estimates. \cite{Hsu:2008} proposed to use the Lasso estimator in VAR models,
\begin{equation}\label{eq:lasso}
\hat a^{\rm Lasso}_{i} := \mathop{{\rm argmin}}\limits_{\alpha\in \mathbb{R}^p}\{ ||Y_i - \mathbf { X } \alpha||_2^2/n + 2\lambda||\alpha||_1\},\quad i=1,...,p,
\end{equation}
where $\lambda$ is the tuning parameter which controls the amount of regularization. In practice, $\lambda$ is often chosen by cross-validation. We denote the set of selected variables by $\hat S_i:=\{j\in \{1,...,p\}:\hat a^{\rm Lasso}_{ij}\ne0\}$ and let $\hat s_{i}:=|\hat S_i|$. The Lasso estimator is hard to  use directly for statistical inference due to its bias. \cite{Zhang:2014} and \cite{vandeGeer:2014}  proposed the de-biased Lasso method for construct confidence intervals, which proceeds as follows.

The Lasso estimator in (\ref{eq:lasso}) satisfies the Karush-Kuhn-Tucker (KKT) conditions,
\begin{equation}\label{eq:KKT-1}
  -\mathbf { X } ^ {\T } ( Y_i - \mathbf { X } \hat { a }^{\rm Lasso}_i ) / n +\lambda \hat { \kappa }=0,
\end{equation}
where $\hat { \kappa }$ is the sub-gradient of $\ell_1$ norm and satisfies $\|\hat { \kappa }\|_\infty\leq 1$ and $\hat { \kappa }_j=\operatorname { sign } ( \hat { a } ^{\rm Lasso}_ { ij } )$ if $\hat { a }^{\rm Lasso} _ {i j } \neq 0$.
By $Y_{i} = \mathbf{X} a_i + {\varepsilon}_{i}$ and $\hat { \Sigma } := \mathbf { X } ^ {\T } \mathbf { X } / n$, we obtain
\begin{equation}\label{eq:KKT-2}
\hat { \Sigma } ( \hat { a }^{\rm Lasso}_i - a_i  ) + \lambda \hat { \kappa } = \mathbf { X } ^ {\rm  T } \varepsilon_i / n.
\end{equation}
With $\Sigma:=E(\hat\Sigma)=E(\mathbf{y}_{t-1}\mathbf{y}_{t-1}^{\T})$, if we have a proper approximation for the precision matrix $\Theta:=\Sigma^{-1}$, say $\hat\Theta$, then multiplying both hand sides of (\ref{eq:KKT-2}) by $\hat\Theta$, we obtain
\begin{equation}\label{eq:KKT-3}
\hat { a }^{\rm Lasso}_i - a_i   + \hat\Theta\lambda \hat { \kappa } = \hat\Theta\mathbf { X } ^ {\rm  T } \varepsilon_i / n + (I-\hat\Theta\hat { \Sigma })(\hat { a }^{\rm Lasso}_i - a_i) .
\end{equation}
The de-biased Lasso estimator is
$$
\hat a_i:=\hat { a }^{\rm Lasso}_i + \hat\Theta\lambda \hat { \kappa } = \hat { a }^{\rm Lasso}_i + \hat\Theta\mathbf { X } ^ {\T } ( Y_i - \mathbf { X } \hat { a }^{\rm Lasso}_i ) / n,
$$
where the second equality is due to (\ref{eq:KKT-1}).
Intuitively, as long as we can prove that the first term of the right-hand of (\ref{eq:KKT-3}) is asymptotically normal with mean zero and the second term is asymptotically negligible, then with a consistent estimator of the variance of $\varepsilon_i$, we can perform inference for $a_i$.

Following \cite{Zhang:2014} and \cite{vandeGeer:2014}, we get $\hat\Theta$ by nodewise Lasso \citep{Meinshausen:2006kb}. Specifically, for $j=1,...,p$, we run a Lasso regression of $X_j$ versus $\mathbf{X}_{-j}$,
\begin{equation}\label{eq:nodelasso}
\hat{\gamma}_j := \mathop{{\rm argmin}}\limits_{\gamma\in \mathbb{R}^{p-1}}\{ ||X_j-\mathbf{X}_{-j}\gamma||_2^2/n + 2\lambda_j||\gamma||_1\},
\end{equation}
where $\lambda_j$ is the tuning parameter. Note that $\hat{\gamma}_j$ with the components of $\{\hat{\gamma}_{jk};k=1,...,p,k\ne j\}$ is an estimator of $\gamma_{j}:=\Sigma_{-j,-j}^{-1}\Sigma_{-j,j}$, where $\Sigma_{-j,-j}$ is the covariance matrix of $\mathbf{X}_{-j}$ and $\Sigma_{-j,j}$ is the covariance matrix of $\mathbf{X}_{-j}$ and $X_j$. We denote the sparsity of $\gamma_{j}$ by $q_j:=|\left\{ k\neq j :  \gamma    _ { jk } \neq 0 \right\}|$. We further denote 
\begin{gather*}
\hat{C}:=   \left[ \begin{matrix}
   1 & -\hat{\gamma}_{12} &  \cdots & -\hat{\gamma}_{1p} \\
   -\hat{\gamma}_{21} & 1 &  \cdots & -\hat{\gamma}_{2p} \\
    \vdots &  \vdots & \ddots & \vdots \\
   -\hat{\gamma}_{p1} & -\hat{\gamma}_{p2} & \cdots &1
  \end{matrix} \right], \\
\hat{\tau}_{j}^{2}:=\left\|X_{j}-\mathbf{X}_{-j} \hat{\gamma}_{j}\right\|_{2}^{2} / n+\lambda_{j}\left\|\hat{\gamma}_{j}\right\|_{1}, \quad \hat { T } ^ { 2 } : = \operatorname { diag } ( \hat { \tau } _ { 1 } ^ { 2 } , \ldots , \hat { \tau } _ { p } ^ { 2 } ).
\end{gather*}
Then we get
\begin{equation}\label{eq:nodelasso3}
\hat { \Theta }  : = \hat { T } ^ { - 2 } \hat { C }.
\end{equation}
\cite{Javanmard:2014} uses a different approach to obtain $\hat\Theta$. We do not adopt their approach because of its inferior finite-sample performance as shown in the simulation section.

According to the discussion in \cite{Reid:2016es}, we estimate the variance of $\varepsilon_i$, say $\sigma^2_i$, by the residual sum of squares of the Lasso estimator divided by its degree of freedom,
\begin{equation}\nonumber
\hat { \sigma } _ { i } ^ { 2 } := \frac { 1 } { n - \hat { s }_{i} } \| \hat \varepsilon_i \| _ { 2 } ^ { 2 },
\end{equation}
where
\begin{equation*}
  \hat \varepsilon_i = (\hat \varepsilon_{i1},...,\hat \varepsilon_{in} )^{\T}:=Y_i-\mathbf { X }\hat { a }^{\rm Lasso}_i.
\end{equation*}

The procedure of nodewise Lasso requires to compute the Lasso solution path for $p$ times, which is the main computation burden of the de-biased Lasso. Fortunately, for different variable $i=1,...,p$, we need only to compute $\hat { \Theta }$ once since the $p$ equations share the same $\mathbf{X}$.
The whole procedures is summarized in Algorithm \ref{alg:dbl}, where every \code{for} loop could be ran in parallel.

\begin{algorithm}[ht]
\caption{Construction of confidence intervals by the de-biased Lasso}
\label{alg:dbl}
	\SetAlgoLined
	\textbf{Input} Data $\{\mathbf{y}_t\}, t=0,...,n$; Confidence level $1-\alpha$.
	
	\textbf{Output} Confidence intervals for elements of transition matrix $A$ in VAR models.
	
	\For {$j=1,...,p$} {
    	Compute the nodewise Lasso estimator $\hat\gamma_j$ and residuals $\hat{Z}_{j}=X_{j}-\mathbf{X}_{-j} \hat{\gamma}_{j}$\;
    }
    Compute the $\hat\Theta$ in (\ref{eq:nodelasso3})\;
    \For {$i=1,...,p$}{   
    	Compute the Lasso estimator $\hat a_i^{\rm Lasso}$ given the data $(Y_i,\mathbf{X})$\;
	    Compute the de-biased Lasso estimator $\hat a_i = \hat { a }^{\rm Lasso}_i + \hat\Theta\mathbf { X } ^ {\T } ( Y_i - \mathbf { X } \hat { a }^{\rm Lasso}_i ) / n$, residuals $\hat \varepsilon_i =Y_i-\mathbf { X }\hat { a }^{\rm Lasso}_i$, and variance estimator $\hat { \sigma } _ { i } ^ { 2 } = \| \hat \varepsilon_i \| _ { 2 } ^ { 2 }/( n - \hat { s }_{i} )$\; 
	    \For {$j=1,...,p$}{
	    	Compute $l_{ij}=\hat a_{ij}-z_{1-\alpha/2} \hat \sigma _ { i } \| \hat Z _ { j } \| _ { 2 } /|\hat Z _ { j } ^ {\T } X _ { j } | $ and $u_{ij}=\hat a_{ij}-z_{\alpha/2} \hat \sigma _ {i} \| \hat Z _ { j } \| _ { 2 }/|\hat Z _ { j } ^ {\T } X _ { j } | $\;  
	    }
    }

    \textbf{Return} $1-\alpha$ confidence interval $[l_{ij},u_{ij}]$ for $a_{ij}$, $i,j=1,...,p$.
\end{algorithm}

\subsection{Bootstrap de-biased Lasso}
\label{sec:resBt}

\begin{algorithm}
\caption{Construction of confidence intervals by bootstrap de-biased Lasso}
\label{alg:bdbl}
	\SetAlgoLined
	\textbf{Input} Data $\{\mathbf{y}_t\}, t=0,...,n$; Confidence level $1-\alpha$; Bootstrap type (residual bootstrap or wild multiplier bootstrap); Number of bootstrap replications $B$.

	\textbf{Output} Confidence intervals for elements of transition matrix $A$ in VAR models.

    \For {$j=1,...,p$} {
    	Compute the nodewise Lasso estimator $\hat\gamma_j$ and residuals $\hat{Z}_{j}=X_{j}-\mathbf{X}_{-j} \hat{\gamma}_{j}$\;
    }
    Compute the $\hat\Theta$ in (\ref{eq:nodelasso3})\;
    \For {$i=1,...,p$} {   
    	Compute the Lasso estimator $\hat a_i^{\rm Lasso}$ given the data $(Y_i,\mathbf{X})$\;
	    Compute the de-biased Lasso estimator $\hat a_i = \hat { a }^{\rm Lasso}_i + \hat\Theta\mathbf { X } ^ {\T } ( Y_i - \mathbf { X } \hat { a }^{\rm Lasso}_i ) / n$, residual $\hat \varepsilon_i =Y_i-\mathbf { X }\hat { a }^{\rm Lasso}_i$, and variance estimator $\hat { \sigma } _ { i } ^ { 2 } = \| \hat \varepsilon_i \| _ { 2 } ^ { 2 }/( n - \hat { s }_{i} )$\; 
	    \For {$b=1,...,B$} {
	    	\uIf {Bootstrap type $==$ residual bootstrap} {
	    		Resample with replacement from the centered residuals $\{\hat \varepsilon_{ij} - \hat\varepsilon_{i\cdot}\}$ and obtain the bootstrap residuals $\varepsilon^*_i=(\varepsilon^*_{i1},...,\varepsilon^*_{in})^{\T}$\;
	    	}
	    	\uElseIf {Bootstrap type $==$ wild multiplier bootstrap} {
	    		Generate i.i.d multiplier $W _ {i 1 } , \ldots , W _ {i n }$ with $E(W _ { ij }) = 0$, $E( W _ { ij } ^ { 2 }) = 1$ and $E (W _ { ij } ^ { 4 }) < \infty$. Then we obtain $\varepsilon _ { ij } ^ { *  } = W _ { ij }( \varepsilon_{ij} - \hat\varepsilon_{i\cdot})$\;
	    	}
	    	Generate bootstrap samples $Y_i ^ { * } = \mathbf { X } \hat {a}_i + \varepsilon _i^ { * }$\;
	    	Compute the Lasso estimator $\hat a_i^{\rm Lasso*}$ given the bootstrap data $(Y_i^*,\mathbf{X})$\;
	    	Compute $\hat a_i^* = \hat { a }^{\rm Lasso*}_i + \hat\Theta\mathbf { X } ^ {\T } ( Y^*_i - \mathbf { X } \hat { a }^{\rm Lasso*}_i ) / n$, $\hat \varepsilon_i^* =Y_i^*-\mathbf { X }\hat { a }^{\rm Lasso*}_i$, and $\hat { \sigma } _ { i } ^ { *2 } = \| \hat \varepsilon^*_i \| _ { 2 } ^ { 2 }/( n - \hat { s }^*_{i} )$\;
	    	Compute the pivot $T^{*(b)}$ with $T^{*(b)}_{ij}=(\hat a_{ij}^*-\hat a^{\rm Lasso*}_{ij})|\hat Z _ { j } ^ {\T } X _ { j } | / (\hat \sigma ^*_ { i} \| \hat Z _ { j } \| _ { 2 })$\;
	    }
	    \For {$j=1,...,p$} {
	    	Let $l_{ij}=\hat a_{ij}-q_{1-\alpha/2} \hat \sigma _ { i } \| \hat Z _ { j } \| _ { 2 } /|\hat Z _ { j } ^ {\T } X _ { j } | $ and $u_{ij}=\hat a_{ij}-q_{\alpha/2} \hat \sigma _ {i} \| \hat Z _ { j } \| _ { 2 }/|\hat Z _ { j } ^ {\T } X _ { j } |$, where $q_{\alpha}$ denotes the lower $\alpha$ quantile of $\{T^{*(1)}_{ij},...,T^{*(B)}_{ij}\}$\;  
	    }
    }
    \textbf{Return} $1-\alpha$ confidence interval $[l_{ij},u_{ij}]$ for $a_{ij}$, $i,j=1,...,p$.
\end{algorithm}

In this section, we introduce the residual bootstrap de-biased Lasso and wild multiplier bootstrap de-biased Lasso proposed in \cite{Dezeure:2017jla}. The only difference of these two bootstrap methods is the approach to generate bootstrap residuals.
For residual bootstrap, we resample with replacement from the centered residuals $\{\hat \varepsilon_{it} - \hat\varepsilon_{i\cdot},t=1,2,...n\}$, where $\hat\varepsilon_{i\cdot}:=\sum_{t=1}^{n}\hat \varepsilon_{it}/n $, and obtain the bootstrap residuals $\varepsilon^*_i=(\varepsilon^*_{i1},...,\varepsilon^*_{in})^{\T}$.
For wild multiplier bootstrap, we generate i.i.d multiplier $W _ {i 1 } , \ldots , W _ {i n }$ with $E(W _ { it }) = 0$, $E( W _ { it } ^ { 2 }) = 1$ and $E (W _ { it } ^ { 4 }) < \infty$, for example, $W_{it}\sim \mathcal{N}(0,1)$,  which is independent of the original data. Then we multiply the centered residuals by multiplier as $\varepsilon _ { it } ^ { *  } = W _ { it }( \varepsilon_{it} - \hat\varepsilon_{i\cdot})$, and obtain the bootstrap residuals $\varepsilon^*_i=(\varepsilon^*_{i1},...,\varepsilon^*_{in})^{\T}$.

Next, the bootstrap sample are generated as
\begin{equation*}
  Y_i ^ { * } = \mathbf { X } \hat {a}_i + \varepsilon _i^ { * }.
\end{equation*}
Then we can replace the original sample $(Y_i,\mathbf { X })$ by the bootstrap sample $(Y^*_i,\mathbf { X })$ to compute the corresponding quantities for the de-biased Lasso. For example, the bootstrap version de-biased Lasso estimator is obtained by
\begin{equation*}
	\hat a_i^* := \hat { a }^{\rm Lasso*}_i + \hat\Theta\mathbf { X } ^ {\T } ( Y_i^* - \mathbf { X } \hat { a }^{\rm Lasso*}_i ) / n,
\end{equation*}
where $\hat { a }^{\rm Lasso*}_i$ is the bootstrap version Lasso estimator and defined by
\begin{equation}
	\hat a^{\rm Lasso*}_{i} := \mathop{{\rm argmin}}\limits_{\alpha\in \mathbb{R}^p}\{ ||Y_i^* - \mathbf { X } \alpha||_2^2/n + 2\lambda^* ||\alpha||_1\},
	\nonumber
\end{equation}
where $\lambda^*$ is the tuning parameter. Note that, one can use different tuning parameters in the original and bootstrap Lasso estimators. Our simulation results indicate that using the same tuning parameters often performs well. We also denote the set of selected variables by $\hat S_i^*:=\{j\in \{1,...,p\}:\hat a^{\rm Lasso*}_{ij}\ne0\}$ and let $\hat s_{i}^*:=|\hat S_i^*|$. 
For both of bootstrap procedures above, we have $E^*(\varepsilon _ { it } ^ { *  } )=0$ and $\sigma^{*2}_i:=E^*(\varepsilon _ { it } ^ { * 2 } )=\sum_{t=1}^{n}(\hat \varepsilon_{it}-\hat\varepsilon_{i\cdot})^2/n$, where $E^*$ indicates the expectation is with respect to the bootstrap measure. 
Similarly, we estimate $\sigma^{*2}_i$ by $\hat { \sigma } _ { i } ^ { *2 } := \| \hat \varepsilon^*_i \| _ { 2 } ^ { 2 }/( n - \hat { s }^*_{i} )$, where $\hat \varepsilon_i^*:=Y_i^*-\mathbf { X }\hat { a }^{\rm Lasso*}_i$.
We repeat the above procedure $B$ times to obtain the empirical distribution of the statistics of interest.

Quantities determined by $\mathbf{X}$, such as $\hat\Theta$, do not need to be re-computed in the bootstrap replications, since $\mathbf{X}$ is the same in every bootstrap sample. This is the most significant advantage of our bootstrap methods over the model-based bootstrap method \citep{Krampe:2018wf}, which will regenerate the entire time series $\{\mathbf{y}_t\}$. Such bootstrap time series will not share the same $\mathbf{X}$, leading to the computations of $\hat\Theta$ for $B$ times, which is usually not feasible in practice for relatively large $p$. Furthermore, we will show that our two bootstrap de-biased Lasso methods are consistent even though they ignore the dependence structure of VAR models.

In all, in order to make the bootstrap de-biased Lasso methods feasible, we save the computation cost of $\hat\Theta$ on two folds. First, once we obtain $\hat\Theta$, we do not need to compute $\hat\Theta$ for the $p$ equations in VAR models; second, we do not need to compute $\hat\Theta$ for $B$ bootstrap replications. The whole procedure is summarized in Algorithm \ref{alg:bdbl}, where every \code{for} loop could be also ran in parallel.

\section{Theoretical results}
\label{sec:theo}

In this section we discuss the theoretical properties of the de-biased Lasso and the bootstrap de-biased Lasso. Different from the fixed design matrix case in linear regression models, the design matrix in equation \eqref{eqn:var1} of VAR models is random, exhibits complex dependence structure, and is correlated with $\varepsilon_i$, that is, equation \eqref{eqn:var1} does not justify the assumptions of linear regression models. However, we can still obtain their asymptotic normality under appropriate conditions, using the deviation bound of the Lasso estimator in high-dimensional sparse VAR models  \citep{Basu:2015ho} and the martingale central limit theorem (Theorem 5.3.4 in \cite{Fuller:1996wa}).

\subsection{Asymptotic distribution of the de-biased Lasso}
\label{sec:dbl-th}

We first introduce the following three assumptions.

\begin{assumption}\label{ap:lambda}
	Suppose that the tuning parameters for the Lasso and nodewise Lasso satisfy: $\lambda\asymp \sqrt{\log(p)/n}$ and $\lambda_j\asymp \sqrt{\log(p)/n}$.
\end{assumption}

\begin{assumption}\label{ap:s0}
	$\max _i s_{i}\log(p)/\sqrt{n}=o (1)$.
\end{assumption}

\begin{assumption}\label{ap:sj}
	$\max_jq_j\log(p)/\sqrt n=o(1)$.
\end{assumption}

Assumption \ref{ap:lambda} requires that the convergence rates of the tuning parameters are of order $\sqrt{\log(p)/n}$. Assumption \ref{ap:s0} is a sparsity assumption on each row of the transition matrix $A$, which is commonly assumed in statistical inference based on the de-biased Lasso methods. These two assumptions are the same as the counterparts in \cite{vandeGeer:2014}. The sparsity assumption on the presicion matrix, Assumption \ref{ap:sj}, is a little stronger than that in \cite{vandeGeer:2014}. We need this assumption because of the complicated dependence structure in VAR models.

\begin{proposition}\label{prop:Riesz}
	Under Assumption \ref{ap:s0}, for $i=1,...,p$, there exist constants $0< c_*<c^*<\infty$, such that, for $s_R>(2 + 2c^*/c_*)s_{i}+1$ and $s_R=O(s_{i})$, the following holds with probability converging to 1:
	\begin{equation}
	c_{*} \leq \min _{\|v\|_0\le s_R} \min _{\|v\|_2=1}\left\|\mathbf{X} v\right\|^{2}_2 / n \leq \max _{\|v\|_0\le s_R} \max _{\|v\|_2=1}\left\|\mathbf{X} v\right\|^{2}_2 / n \leq c^{*}.
	\nonumber
	\end{equation}
\end{proposition}
Proposition \ref{prop:Riesz} states that $\mathbf{X}$ satisfies the sparse Riesz condition \citep{Zhang:2008ga}, which bounds the extreme eigenvalues of $\mathbf{X}^\T\mathbf{X}/n$ in a sparse space. \cite{Basu:2015ho} has obtained the lower bound in the sparse Riesz condition and we complement their result by providing the upper bound. The sparse Riesz condition is crucial for proving the conclusion (b) of Theorem \ref{thm:Lasso} below.
\begin{theorem}\label{thm:Lasso}
For $i=1,...,p$, \\
(a) Under Assumptions \ref{ap:lambda} and \ref{ap:s0}, we have
\begin{gather}
	|| \hat a^{\rm Lasso}_{i}- a_i  || _ { 1 } = O_p(s_{i}\sqrt{\log (p)/n}), \nonumber \\
	||\mathbf{X}(\hat a^{\rm Lasso}_{i}- a_i )||_2^2/n=O_p(s_{i}\log (p)/n). \nonumber
\end{gather}
(b) Under Assumption \ref{ap:s0}, for any $\lambda\geq 4C_1c^*/c_*\sqrt{\log p/n}$, where $C_1$ is a constant defined in the supplementary material, we have
$$
\hat s_{i}=O_p( s_{i} ).
$$
\end{theorem}

The first statement of Theorem~\ref{thm:Lasso} provides the estimation and prediction error bounds for the Lasso estimator, which has been established by \cite{Basu:2015ho}. 
The second statement (b) provides an upper bound on the sparsity of the Lasso estimator. This bound has been obtained under high-dimensional sparse linear regression models \citep{Zhang:2008ga}. Theorem~\ref{thm:Lasso} (b) extends the result to high-dimensional sparse VAR models. It is essential for proving the consistency of the variance estimator in the following Theorem \ref{thm:dbl} and for showing the validity of the bootstrap de-biased Lasso, especially the latter, since it requires the sparsity of $\hat a^{\rm Lasso}_{i}$. Propositions 4.1 in \cite{Basu:2015ho} provided a similar bound, but for a thresholded variant of the Lasso.

\begin{proposition}\label{prop:nodelasso}
Under Assumptions \ref{ap:lambda} and \ref{ap:sj}, for $j=1,...,p$, we have
\begin{gather}
	\| \hat { \gamma } _ { j } - \gamma _ { j } \| _ { 1 } = O_p \left(q _ { j }\sqrt{\log p/n}\right),\nonumber \\
	\| \mathbf { X } _ { - j } \left( \hat { \gamma } _ { j } - \gamma _ { j } \right) \| _ { 2 } ^ { 2 } / n = O_p \left(q _ { j } { \log p/n}\right).\nonumber
\end{gather}
\end{proposition}
Proposition \ref{prop:nodelasso} provides the estimation and prediction error bounds for the nodewise Lasso estimator defined in (\ref{eq:nodelasso}).

\begin{theorem}\label{thm:dbl}
Under Assumptions \ref{ap:lambda}, \ref{ap:s0} and \ref{ap:sj}, the de-biased Lasso estimator is asymptotically normal, that is,
\begin{equation*}
  (\hat a_{ij}- a_{ij}) /s.e._{ij}\stackrel { d } { \rightarrow }  \mathcal{N}(0,1), \quad i,j=1,...,p,
\end{equation*}
where
\begin{equation*}
	 s.e. _{ij} =  \frac { \sigma _ {i} \| \hat Z _ { j } \| _ { 2 } } {|\hat Z _ { j } ^ {\T } X _ { j } | }.
\end{equation*}
Furthermore,
\begin{equation*}
	\hat\sigma _ { i }/\sigma _ {i }\stackrel { P } { \rightarrow } 1,\quad i=1,...,p.
\end{equation*}
\end{theorem}

\begin{remark}
The asymptotic normality of the de-biased Lasso is also proven in Theorem 3.4 in \cite{Zheng:2018}. They propose $\sum_{i=1}^p\| \hat \varepsilon_i \| _ { 2 } ^ { 2 }/ (np)$ as a consistent variance estimator of $\varepsilon_i$.
\end{remark}

Theorem \ref{thm:dbl} shows that the de-biased Lasso estimator is asymptotically normal and its asymptotic variance can be estimated consistently. Thus, we can construct an asymptotically valid confidence interval for each element of $A$, $a_{ij}$, by using normal approximation.

\subsection{Asymptotic distribution of the bootstrap de-biased Lasso}
\label{sec:bdbl-th}

\begin{assumption}\label{ap:lambda_bt}
	Suppose that the tuning parameters for the Lasso, nodewise Lasso, and bootstrap Lasso satisfy: $\lambda \asymp \sqrt{\log(p)/n}$, $\lambda_j \asymp \sqrt{\log(p)/n}$, and $\lambda^* \asymp \log p/\sqrt{n}$.
\end{assumption}

\begin{assumption}\label{ap:s0_bt}
	$s_{i}(\log p)^{3/2}/\sqrt{n}=o (1), i=1,...,p$.
\end{assumption}

Assumptions~\ref{ap:lambda_bt} and \ref{ap:s0_bt} require stronger convergence rates compared to those in \cite{Dezeure:2017jla} used to obtain the asymptotic distribution of the bootstrap de-biased Lasso in high-dimensional sparse linear regression models. However, \cite{Dezeure:2017jla} assumed that $\|\mathbf{X}\|_\infty=O(1)$, which does not hold for the random design matrix in VAR models. In fact, we can only show that $\|\mathbf{X}\|_\infty=O_p(\sqrt{\log pn} )$. This is why we require stronger conditions on the sparsity and  tuning parameters of the bootstrap Lasso estimator.

\begin{theorem}\label{thm:BtLasso}
For $i=1,...,p$, \\
(a) Under Assumptions \ref{ap:lambda_bt} and \ref{ap:s0_bt}, we have
\begin{gather}
	||\hat a^{\rm Lasso*}_{i}-\hat { a }^{\rm Lasso}_{i} || _ { 1 } = O_p(s_{i}\log (p)/\sqrt{n}), \nonumber \\
	||\mathbf{X}(\hat a^{\rm Lasso*}_{i}-\hat { a }^{\rm Lasso}_{i} )||_2^2/n=O_p(s_{i}\log^2 (p)/n). \nonumber
\end{gather}
(b) Under Assumption \ref{ap:s0_bt}, for any $\lambda^*\geq 4C_2 c^*/c_*\sqrt{\log (p)/n}$, where $C_2$, $c^*$ and $c_*$ are constants defined in the supplementary material, we have
$$
\hat s^*_{i}=O_p( s_{i} ),
$$
where $\hat s^*_{i}=|\{j\in \{1,...,p\}:\hat a^{\rm Lasso*}_{ij}\ne0\}|$.
\end{theorem}

Theorem~\ref{thm:BtLasso} is the bootstrap analogue of Theorem~\ref{thm:Lasso}. The different convergence rates in statement (a) are due to Assumption \ref{ap:lambda_bt}.

\begin{theorem}\label{thm:bdbl}
Under Assumptions \ref{ap:sj} -- \ref{ap:s0_bt}, the bootstrap de-biased Lasso estimators are asymptotically normal, that is,
\begin{equation*}
  (\hat{a}_{ij}^*-\hat { a }^{\rm Lasso}_{ij} ) /s.e.^*_{ij}\stackrel { d^* } { \rightarrow }  \mathcal{N}(0,1) \; \text{in probability},\quad i,j=1,...,p,
\end{equation*}
where $\hat{a}_{ij}^*$ denotes either residual or multiplier wild bootstrap de-biased Lasso estimators,
\begin{equation*}
	 s.e. ^*_{ij} =  \frac { \sigma ^*_ { i } \| \hat Z _ { j } \| _ { 2 } } {|\hat Z _ { j } ^ {\T } X _ { j } | },
\end{equation*}
and $d^*$ indicates the convergence is with respect to the bootstrap measure. Furthermore,
\begin{equation*}
	\hat\sigma^* _ { i }/\sigma^* _ { i }\stackrel { P ^*} { \rightarrow } 1\; \text{in probability},\quad i=1,...p.
\end{equation*} 
\end{theorem}

Theorems \ref{thm:dbl} and \ref{thm:bdbl} imply that the conditional distributions of both residual and multiplier wild bootstrap de-biased Lasso estimators are valid approximations to the (unconditional) distribution of the de-biased Lasso estimator. Thus, we could perform valid inference about $a_{ij}$'s using the bootstrap.

\section{Simulation studies}
\label{sec:sim}

We evaluate the finite-sample performance of the proposed methods by simulation studies in this section. 
We compare the methods with the bootstrap Lasso/Lasso+OLS and another de-biased Lasso method, JM \citep{Javanmard:2014}, in terms of bias, root mean squared error (RMSE), coverage probabilities and mean confidence interval lengths.

We use \code{R} package \code{hdi} to implement the de-biased Lasso, residual bootstrap de-biased Lasso and multiplier wild bootstrap de-biased Lasso, \code{R} package \code{HDCI} to implement the bootstrap Lasso and bootstrap Lasso+OLS, and \code{R} code provided by \cite{Javanmard:2014} to implement their version of de-biased Lasso. The tuning parameters are selected by 10-fold cross-validation. We set the number of bootstrap replications $B=500$.

\subsection{Setups}

With sample size $n= 100,300$, dimension $p=200$ and sparsity $s_{i}=5, 10$ for $i=1,...,p$, we first generate transition matrix $A$ as follows.
\begin{enumerate}
  \item We generate $s_{i}\times p$ non-zero parameters independently from a uniform distribution on $[-1,-0.5]\cup[0.5,1]$;
  \item On each row of $A^{init}$, we set the diagonal element to be non-zero and randomly arrange the other $s_{i}-1$ non-zero parameters on the other positions;
  \item Since the largest modulus of eigenvalues of $A^{init}$ may greater than 1, in order to make time series stable, we let
$$
A=\frac{0.9}{\Lambda_{\rm max}(A^{init})}{A^{init}}.
$$
\end{enumerate}

The third step makes the maximum eigenvalue of $A$ equals to $0.9$. We also try to let  $\Lambda_{\rm max}(A)=0.7$ and the results are similar.
Unlike in linear regression models we can arbitrarily set the range of the absolute value of the non-zero parameters, the third step will make the absolute value of each parameter small, no matter how large the non-zero parameters are generated in the first step. In our simulation, when $s_{i}$ equals $5$ and $10$, the ranges of the absolute values of elements of $A$ are $(0.23,0.46)$ and $(0.17,0.34)$ respectively. The smaller absolute values of the parameters implies that the de-biased Lasso will outperform the bootstrap Lasso/Lasso+OLS since the validity of the latter usually requires the ``beta-min'' condition (all nonzero parameters are sufficiently large in absolute values) while the former does not.

We generate data $\{\mathbf{y}_1,...,\mathbf{y}_n\}$ from VAR model (\ref{eq:VAR}) with
\begin{enumerate}
	\item homoscedastic Gaussian errors $\mathbf{u}_t = \mathbf{\xi}_t$;
	\item homoscedastic non-Gaussian errors $\mathbf{u}_t = (\xi_t^2 - 1)/\sqrt{2}$;
	\item heteroscedastic Gaussian errors $\mathbf{u}_t = \eta_t\,\xi_t$;
	\item heteroscedastic non-Gaussian errors $\mathbf{u}_t = \eta_t(\xi_t^2 - 1)/\sqrt{2}$,
\end{enumerate}
where $\mathbf{\xi}_t \stackrel{i.i.d} \sim\mathcal{N}_p (\mathbf{0},I)$ and $\eta_t\stackrel{i.i.d}\sim U(1,3)$. Note that the first type of errors satisfies our theoretical assumptions while the other three are not, which are used to explore the robustness of our methods to the distributions of errors.

For homoscedastic Gaussian errors with different $n$ and $s_{i}$, there are four cases in total and the results will be showed in Section \ref{sec:homo-sim-est} and Section \ref{sec:homo-sim-ci}.
For the other three types of errors, since we set $n=100$ and $s_{i}=5$, there are three cases in total and the results will be discussed in Section \ref{sec:heter-sim}.
For every case, we generate 1000 sets of data in order to evaluate the repeated sampling performance of different methods.
We only report the results with respect to the first row of $A$ since the conclusions for other rows are similar.

\begin{figure}[t]
\centering
\includegraphics[scale=0.8]{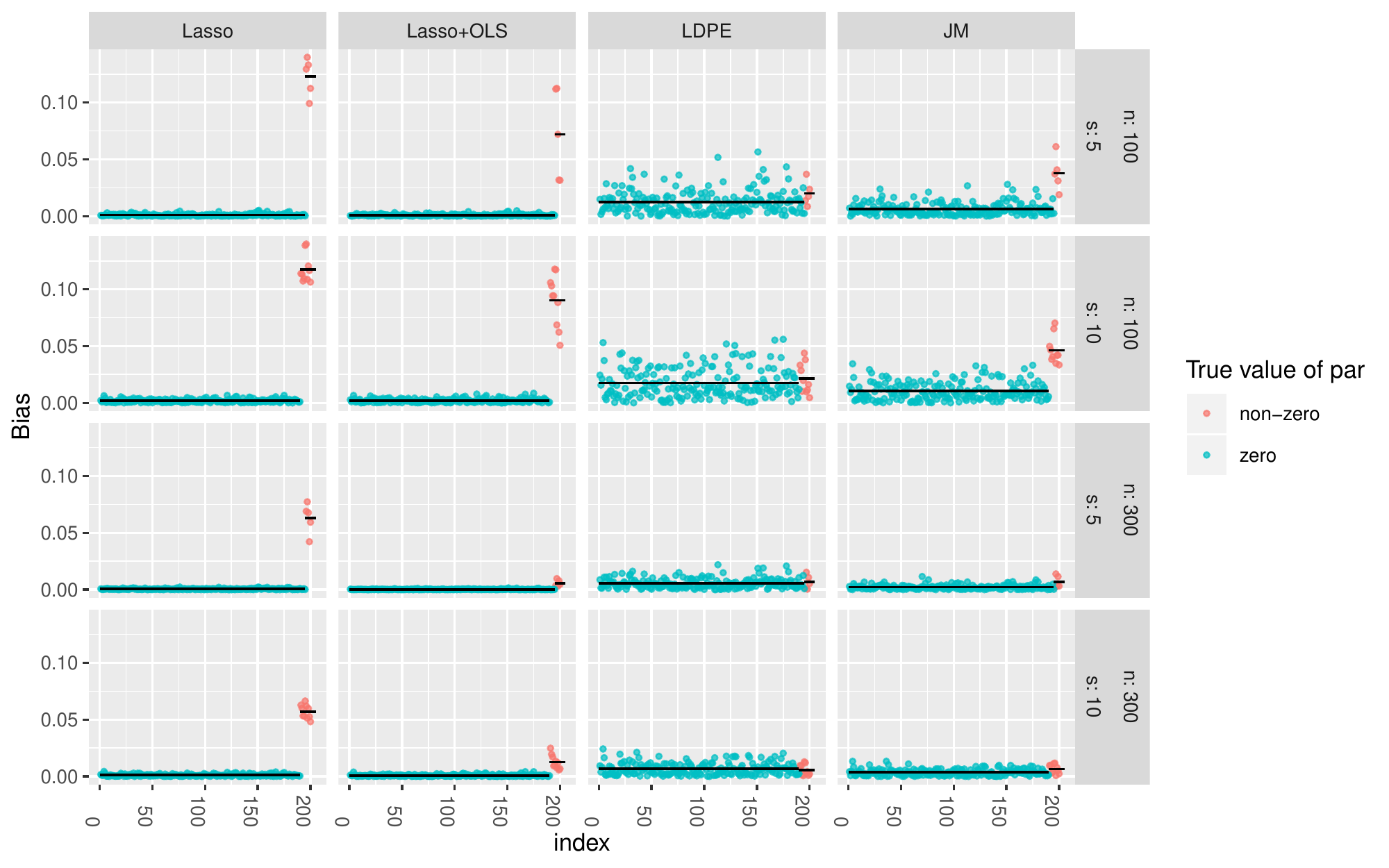}\par
\caption{Comparison of absolute bias for 1000 replications produced by four methods (columns) in four cases (rows). Index on the $x$-axis corresponds to different $a_{1j}$'s, which are arranged from small to large in absolute values. The first $p-s_1$ elements of $a_{1j}$'s are zeros (blue points) and the last $s_1$ are non-zeros (red points). The black lines are total averages of absolute bias for zero and non-zero parameters respectively.}
\label{fig:bias}
\end{figure}

\begin{figure}[t]
\centering
\includegraphics[scale=0.8]{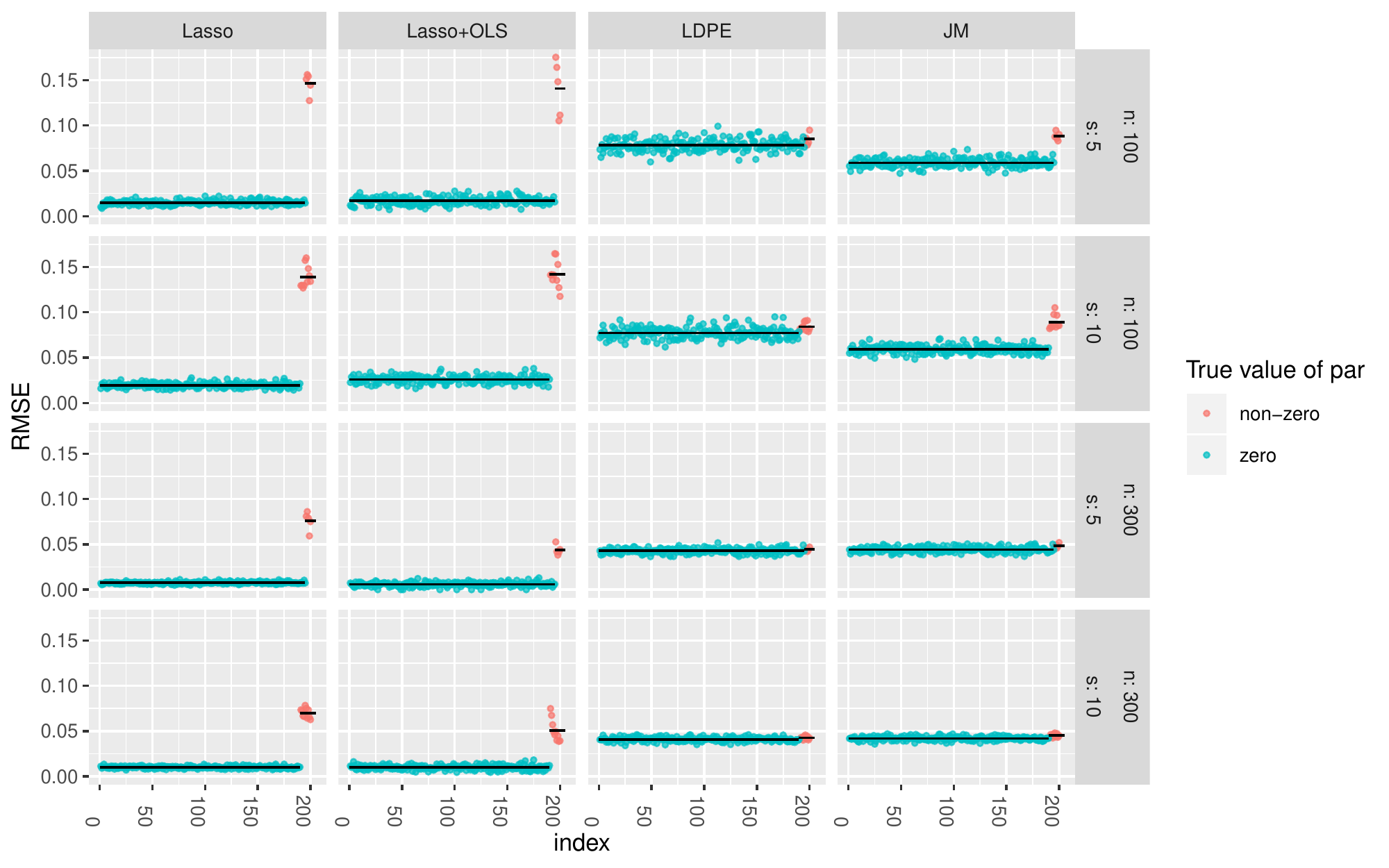}\par
\caption{Comparison of RMSE for 1000 replications produced by four methods (columns) in four cases (rows). Index on the $x$-axis corresponds to different $a_{1j}$'s, which are arranged from small to large in absolute values. The first $p-s_1$ elements of $a_{1j}$'s are zeros (blue points) and the last $s_1$ are non-zeros (red points). The black lines are total averages of absolute RMSE for zero and non-zero parameters respectively.}
\label{fig:RMSE}
\end{figure}

\subsection{Comparison of bias and RMSE}
\label{sec:homo-sim-est}

\begin{table}[ht]
  \centering
  \caption{Average absolute bias and RMSE}
  \label{tab:bias}
    \begin{tabular}{rrlrrrr}
    \hline
    \multicolumn{1}{l}{n} & \multicolumn{1}{l}{s} & $a_{ij}$ & \multicolumn{1}{l}{Lasso} & \multicolumn{1}{l}{Lasso+OLS} & \multicolumn{1}{l}{LDPE} & \multicolumn{1}{l}{JM} \\
    \hline
    &&&\multicolumn{4}{c}{Average absolute bias}\\
    \cline{4-7}
    100   & 5     & non-zero & 0.1228 & 0.0719 & 0.02  & 0.0378 \\
    100   & 5     & zero  & 0.001 & 0.0009 & 0.0122 & 0.0062 \\
    100   & 10    & non-zero & 0.1174 & 0.0902 & 0.0215 & 0.0463 \\
    100   & 10    & zero  & 0.0019 & 0.0021 & 0.0175 & 0.0104 \\
    300   & 5     & non-zero & 0.0631 & 0.0058 & 0.0066 & 0.0068 \\
    300   & 5     & zero  & 0.0006 & 0.0003 & 0.0056 & 0.0022 \\
    300   & 10    & non-zero & 0.057 & 0.0126 & 0.0056 & 0.0066 \\
    300   & 10    & zero  & 0.0011 & 0.0008 & 0.0067 & 0.0036 \\
    \hline
    &&&\multicolumn{4}{c}{Average RMSE}\\
    \cline{4-7}
    100   & 5     & non-zero & 0.1467 & 0.1409 & 0.0853 & 0.0882 \\
    100   & 5     & zero  & 0.0149 & 0.017 & 0.0783 & 0.0589 \\
    100   & 10    & non-zero & 0.1389 & 0.1422 & 0.0842 & 0.0888 \\
    100   & 10    & zero  & 0.0195 & 0.0258 & 0.0772 & 0.0592 \\
    300   & 5     & non-zero & 0.0761 & 0.0437 & 0.0446 & 0.0485 \\
    300   & 5     & zero  & 0.0076 & 0.0058 & 0.0428 & 0.0443 \\
    300   & 10    & non-zero & 0.0698 & 0.0507 & 0.0425 & 0.0451 \\
    300   & 10    & zero  & 0.0101 & 0.01  & 0.0409 & 0.0419 \\
    \hline
    \end{tabular}
  \label{tab:est}
\end{table}

In this section, we compare the bias and RMSE of four estimation methods: Lasso, Lasso+OLS, LDPE and JM.
Figures \ref{fig:bias} and \ref{fig:RMSE} and Table \ref{tab:bias} show the results of absolute bias ($|E\hat a_{1j}-a_{1j}|$) and RMSE ($[E(\hat a_{1j}-a_{1j})^2]^{1/2}$).
For non-zero parameters, the Lasso estimator has large bias, the Lasso+OLS estimator reduces the bias (23\% - 41\% when $n=100$, 77\% - 90\% when $n=300$), and two de-biased Lasso estimators further reduce the bias (60\% - 90\%).
For zero parameters, the Lasso and Lasso+OLS estimators have nearly zero bias, while the bias of two de-biased Lasso estimators are about the same magnitude as those of non-zero parameters.
In terms of RMSE, when $n=100$, situation is almost the same as bias; when $n=300$, for those non-zero parameters,  the Lasso and Lasso+OLS estimators have RMSE comparable to two de-biased Lasso estimators while for those zero parameters, the Lasso and Lasso+OLS estimators have much smaller RMSE. Specifically, for those zero parameters, compared to the LDPE, the Lasso reduces the RMSE by 75\% - 82\%.
For estimation purpose, we recommend the Lasso and Lasso+OLS since their RMSEs are smaller.
However, for construction of confidence intervals, small bias will lead to more accurate coverage probabilities, which will be seen in the next section.

\subsection{Comparison of coverage probabilities and confidence interval lengths}
\label{sec:homo-sim-ci}

\begin{figure}[ht]
\centering
\includegraphics[scale=0.8]{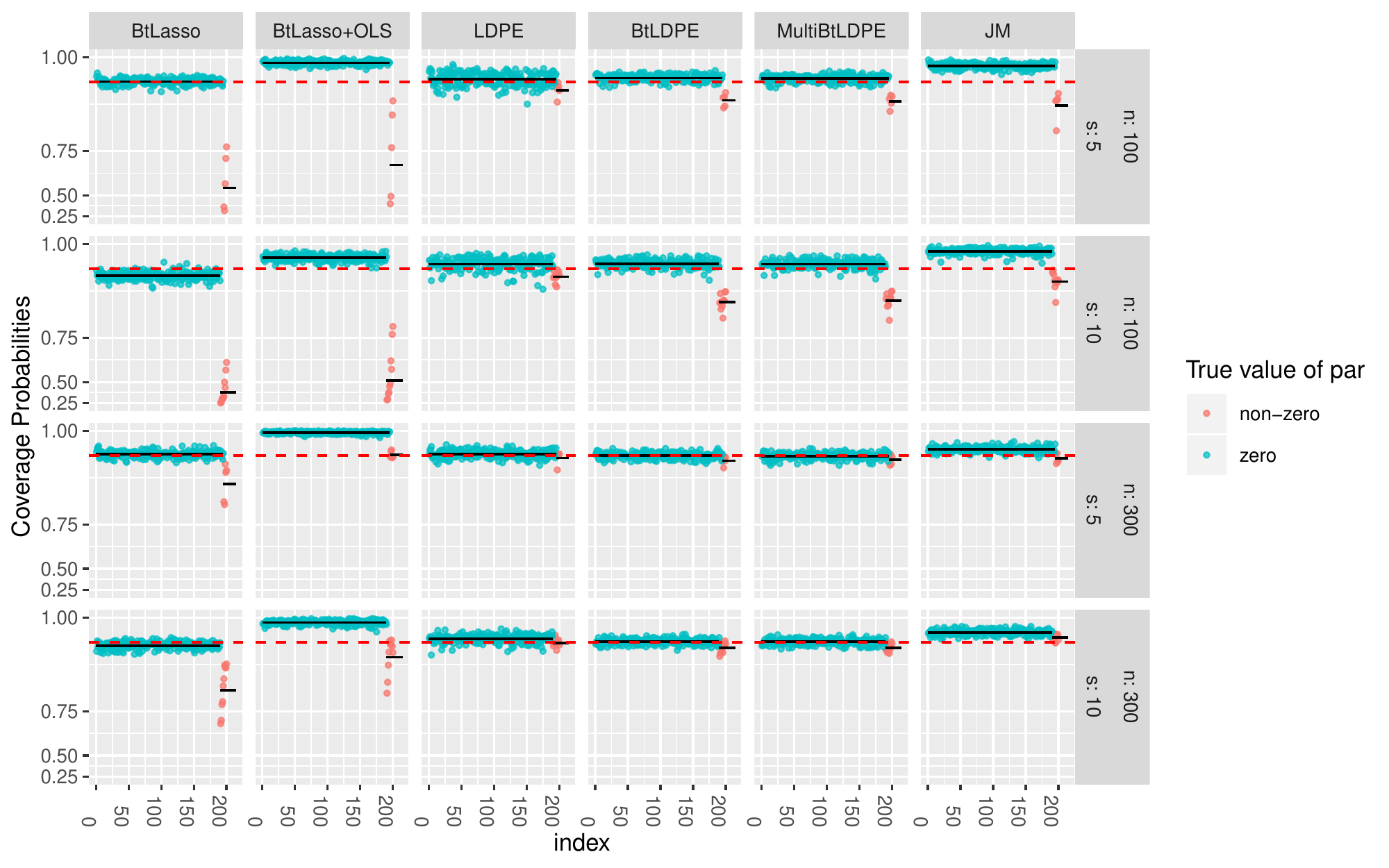}\par
\caption{Comparison of empirical coverage probabilities for 1000 replications produced by six methods (columns) in four cases (rows). Index on the $x$-axis corresponds to different $a_{1j}$'s, which are arranged from small to large in absolute values. The first $p-s_1$ elements of $a_{1j}$'s are zeros (blue points) and the last $s_1$ are non-zeros (red points). The black lines are total averages of coverage probabilities for zero and non-zero parameters respectively. The red dashed lines correspond to the nominal confidence level 95\%.}
\label{fig:cp1}
\end{figure}

\begin{figure}[ht]
\centering
\includegraphics[scale=0.8]{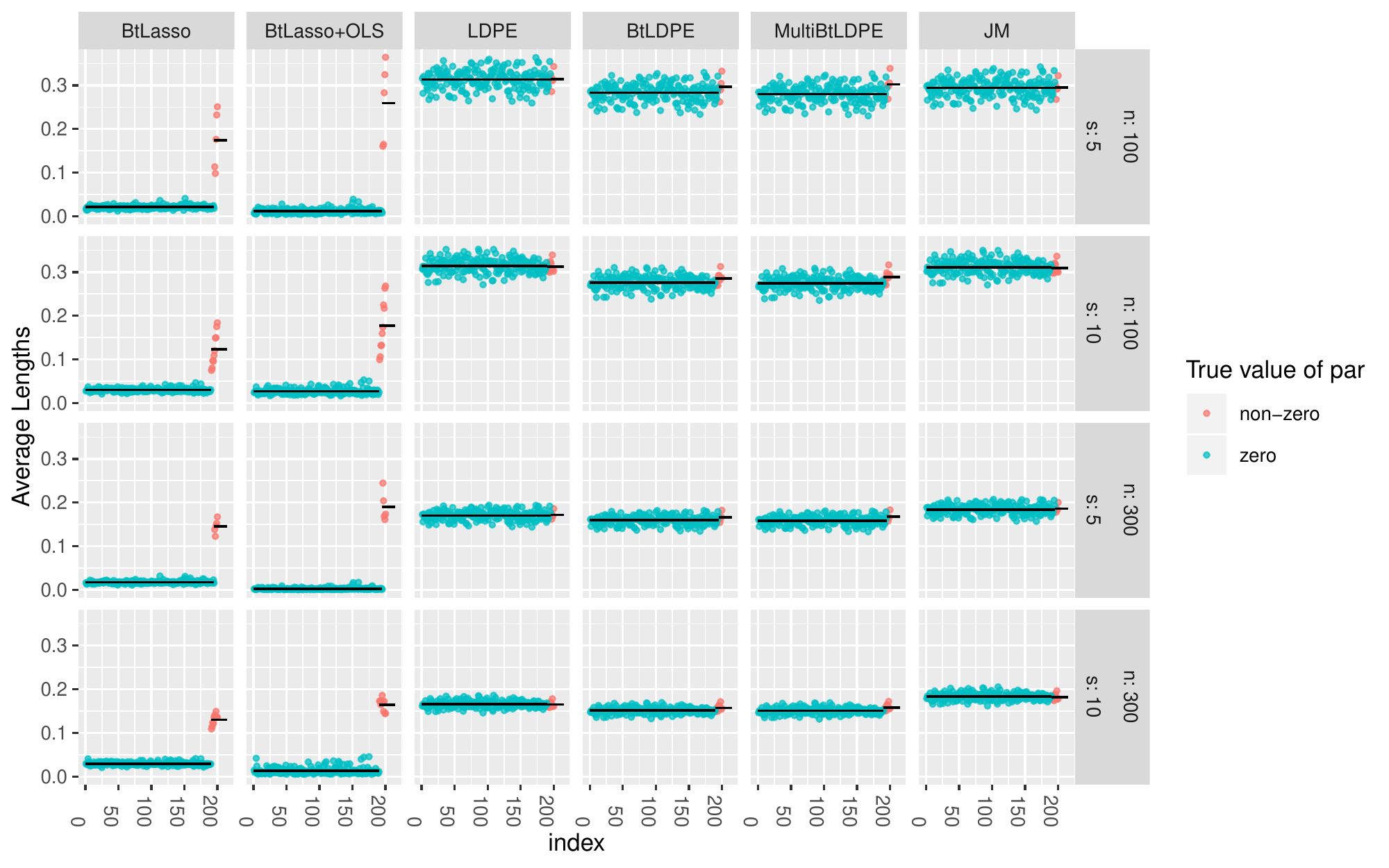}\par
\caption{
Comparison of average confidence interval lengths for 1000 replications produced by six methods (columns) in four cases (rows). Index on the $x$-axis corresponds to different $a_{1j}$'s, which are arranged from small to large in absolute values. The first $p-s_1$ elements of $a_{1j}$'s are zeros (blue points) and the last $s_1$ are non-zeros (red points). The black lines are total averages of interval lengths for zero and non-zero parameters respectively.
}
\label{fig:length1}
\end{figure}

\begin{table}[ht]
  \centering
  \caption{Average empirical coverage probabilities and average confidence interval lengths}
  \label{tab:coverage}
    \begin{tabular}{rrlrrrrrr}
    \hline
    \multicolumn{1}{l}{n} & \multicolumn{1}{l}{s} & $a_{ij}$ & \multicolumn{1}{l}{BtLasso} & \multicolumn{1}{l}{BtLassoOLS} & \multicolumn{1}{l}{LDPE} & \multicolumn{1}{l}{BtLDPE} & \multicolumn{1}{l}{MultiBtLDPE} & \multicolumn{1}{l}{JM} \\
    \hline
    &&&\multicolumn{6}{c}{Average empirical coverage probabilities (\%)}\\
    \cline{4-9}
    100   & 5     & non-zero & 55.8  & 69    & 93.2  & 90.7  & 90.5  & 89.4 \\
    100   & 5     & zero  & 95.1  & 99    & 95.6  & 95.9  & 95.7  & 98.3 \\
    100   & 10    & non-zero & 40    & 51.3  & 93.2  & 86.9  & 87.2  & 92.1 \\
    100   & 10    & zero  & 93.4  & 97.3  & 96    & 96.1  & 96    & 98.6 \\
    300   & 5     & non-zero & 88.2  & 95.2  & 94.5  & 93.8  & 94.1  & 94.3 \\
    300   & 5     & zero  & 95.3  & 99.7  & 95.4  & 95    & 94.9  & 96.3 \\
    300   & 10    & non-zero & 82.4  & 91.5  & 94.8  & 93.7  & 93.8  & 96 \\
    300   & 10    & zero  & 94.3  & 99    & 95.7  & 95.1  & 95.1  & 97.1 \\
    \hline
    &&&\multicolumn{6}{c}{Average confidence interval lengths}\\
    \cline{4-9}
    100   & 5     & non-zero & 0.174 & 0.259 & 0.315 & 0.296 & 0.302 & 0.295 \\
    100   & 5     & zero  & 0.021 & 0.012 & 0.313 & 0.283 & 0.28  & 0.294 \\
    100   & 10    & non-zero & 0.124 & 0.178 & 0.313 & 0.285 & 0.289 & 0.309 \\
    100   & 10    & zero  & 0.03  & 0.027 & 0.314 & 0.276 & 0.274 & 0.311 \\
    300   & 5     & non-zero & 0.145 & 0.19  & 0.172 & 0.166 & 0.167 & 0.186 \\
    300   & 5     & zero  & 0.017 & 0.003 & 0.17  & 0.159 & 0.159 & 0.184 \\
    300   & 10    & non-zero & 0.13  & 0.165 & 0.165 & 0.157 & 0.158 & 0.182 \\
    300   & 10    & zero  & 0.03  & 0.013 & 0.166 & 0.151 & 0.151 & 0.183 \\
    \hline
    \end{tabular}%
  \label{tab:ci}%
\end{table}%

We now compare the coverage probabilities and mean confidence interval lengths of 95\% confidence intervals constructed by six methods: de-biased Lasso (LDPE), residual bootstrap de-biased Lasso (BtLDPE), multiplier wild bootstrap de-biased Lasso (MultiBtLDPE), bootstrap Lasso (BtLasso), bootstrap Lasso+OLS (BtLasso+OLS) and de-biased Lasso of \cite{Javanmard:2014} (JM).

Figure \ref{fig:cp1} and Table \ref{tab:coverage} show the results of coverage probabilities. 
For non-zero parameters, the coverage probabilities of BtLasso and BtLasso+OLS do not reach the nominal confidence level in all cases except for $n=300, s_{i}=5$. The LDPE, BtLDPE, MultiBtLDPE and JM can reach the nominal confidence level when $n=300$ while only the LDPE reaches the nominal confidence level when $n=100$.
For zero parameters, all methods except the BtLasso reach the nominal confidence level. Note that the BtLasso+OLS and JM produce much higher coverage probabilities, for example 99\%, than the nominal level 95\%.

Figure \ref{fig:length1} and Table \ref{tab:coverage} show the results of average confidence interval lengths.
For non-zero parameters, the BtLasso and BtLasso+OLS produce shorter confidence intervals than the other four methods. Moreover, compared to the LDPE, the BtLDPE reduces confidence interval lengths by 6\% - 12\% when $n=300$ and 3\% - 9\% when $n=100$.
For zero parameters, the BtLasso and BtLasso+OLS have nearly zero average confidence interval lengths, reflecting the super-efficiency of these two methods. For the other four methods, the LDPE, BtLDPE, MultiBtLDPE and JM, the confidence interval lengths for zero parameters are nearly the same as those for non-zero parameters.
Meanwhile, confidence intervals produced by the BtLDPE and MultiBtLDPE are shorter than those produced by the LDPE and JM.

Taking into account both coverage probabilities and confidence interval lengths, when $n$ is small, we recommend the LDPE for its honest coverage probabilities; when $n$ is large, we recommend the BtLDPE and MultiBtLDPE for their honest coverage probabilities and shorter confidence interval lengths.

\subsection{Robustness to the distributions of  errors}
\label{sec:heter-sim}

In this subsection, we explore the robustness of our methods to different distributions of errors, namely, homoscedastic non-Gaussian errors, heteroscedastic Gaussian errors and heteroscedastic non-Gaussian errors.
Compared to homoscedastic Gaussian errors, different distributions of  errors do not lead to significant difference of performance; see the results in the supplementary material. 
Again, we can see that the LDPE has honest coverage probabilities and the BtLDPE and MultiBtLDPE have shorter confidence interval lengths compared to the LDPE.

\section{Real data}

\begin{figure}[ht]
\centering
\includegraphics[scale=0.18]{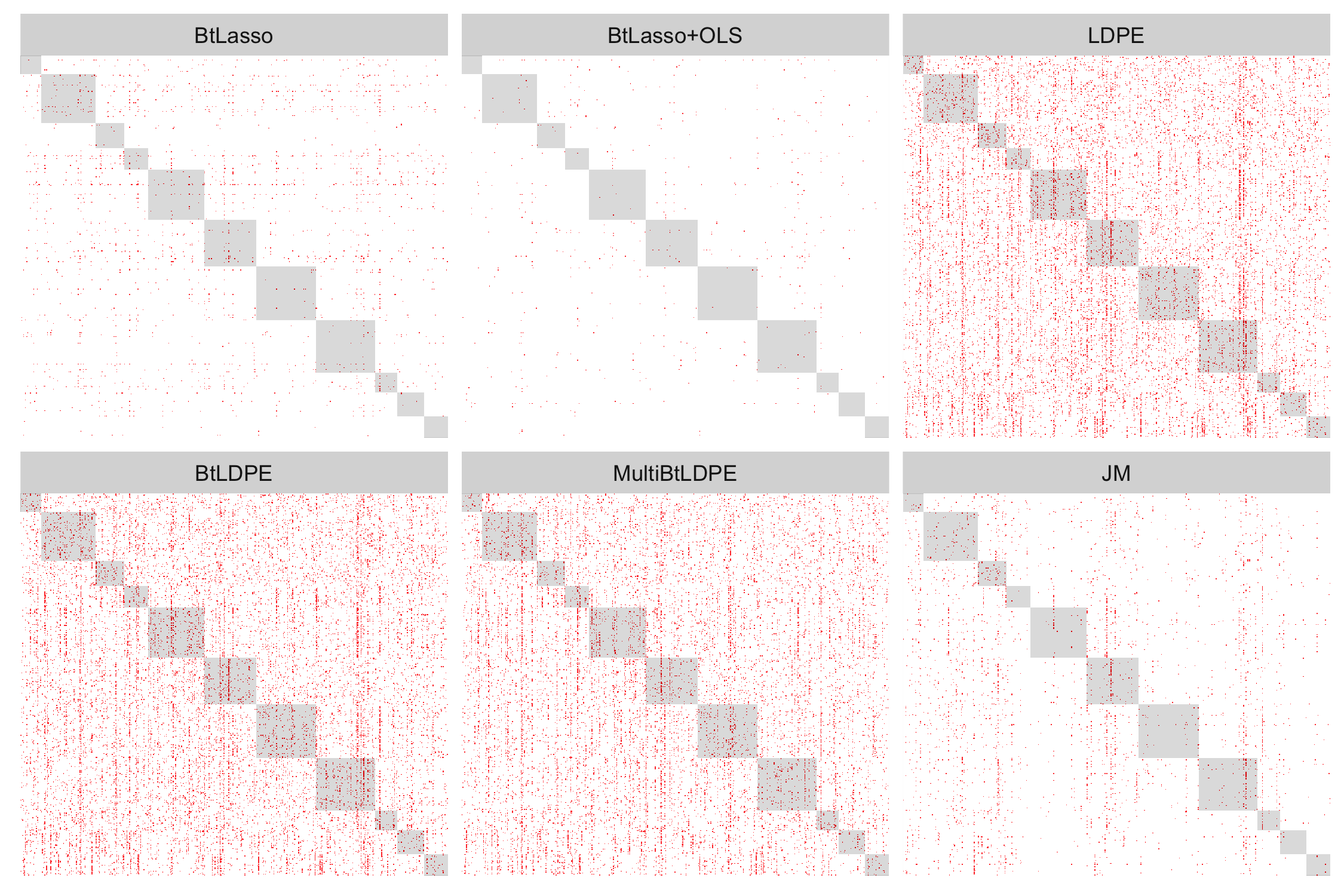}\par
\caption{
Statistical inference results of the $500 \times 500$ transition matrix of the VAR model for the returns of the S\&P 500 constituent stocks using six methods: de-biased Lasso (LDPE), residual bootstrap de-biased Lasso (BtLDPE), multiplier wild bootstrap de-biased Lasso (MultiBtLDPE), de-biased Lasso of \cite{Javanmard:2014} (JM), bootstrap Lasso (BtLasso) and bootstrap Lasso+OLS (BtLasso+OLS). 
The red point indicates that the corresponding parameter $a_{ij}$ is significant (its 95\% confidence interval does not include 0).
The gray square indicates that the two stocks corresponding to $a_{ij}$ are in the same sector. 
}
\label{fig:sp500-sig}
\end{figure}

\begin{figure}[ht]
\centering
\includegraphics[scale=0.18]{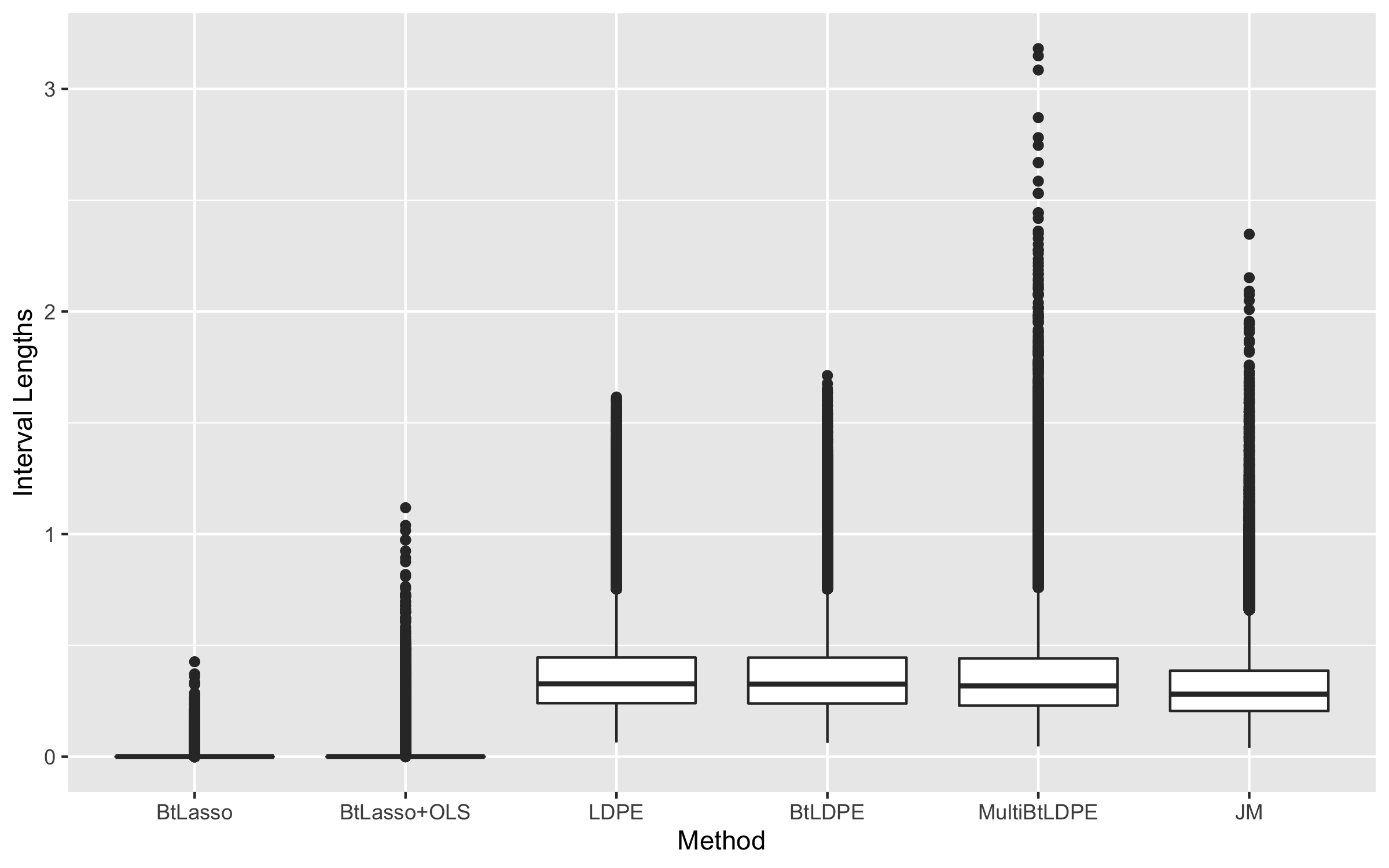}\par
\caption{
Box plot of confidence interval lengths of the $500\times 500$ parameters produced by six methods: de-biased Lasso (LDPE), residual bootstrap de-biased Lasso (BtLDPE), multiplier wild bootstrap de-biased Lasso (MultiBtLDPE), de-biased Lasso of \cite{Javanmard:2014} (JM), bootstrap Lasso (BtLasso) and bootstrap Lasso+OLS (BtLasso+OLS).
}
\label{fig:sp500-length}
\end{figure}

In practice, researchers often use VAR models to analyze the time series data of stocks and conduct statistical inference on the elements of transition matrix, so as to produce knowledge about the relationship between different stocks.
In this section, we use the prices of the S\&P 500 constituent stocks in 2019 to demonstrate our methods.
There are 505 stocks in total because there are five companies have two share classes of stock. There are 252 trading days in 2019, but five stocks have incomplete data for some reasons. After deleting these five stocks, we have data of 500 stocks for 252 days. 
Because our model needs the stationarity of time series data, we use the daily return 
$$
r_{t}=\frac{p_{t}}{p_{t-1}}-1, \quad t=0, \ldots, n
$$
as $\{\mathbf{y}_t\}$ in model (\ref{eq:VAR}), where $p_{t}$ is the adjusted price. Since the transformation reduces one observation and $t$ begins from $0$, in our model, $n=250$.
These 500 companies are in eleven different sectors: Communication Services, Consumer Discretionary, Consumer Staples, Energy, Financials, Health Care, Industrials, Information Technology, Materials, Real Estate and Utilities.

We apply the same six methods as in Section \ref{sec:sim} to this data set and obtain the 95\% confidence intervals for elements of the transition matrix $A$.
The results are shown in Figures \ref{fig:sp500-sig} and \ref{fig:sp500-length}.
First, compared with the BtLasso, BtLasso+OLS and JM, our proposed methods produce more significant parameters, which can provide users with more candidates for effective relationships to make future decisions.
Second, some of the columns of the estimated transition matrix have many significant parameters (see the vertical red lines in Figure \ref{fig:sp500-sig}), indicating that the prices of some stocks have prediction power on the prices of most stocks. A further look at the results reveals that all six methods indicate that Newmont Corporation is such a stock having the ability to predict many other stock prices in advance. Newmont Corporation is the largest producer of gold in the world and the only gold producer listed in the S\&P 500 Index. Considering that the turbulent global financial environment in 2019 made gold the best safe-haven asset, our findings have practical significance.
Third, there are relationships both within and between sectors.
Finally, Figure \ref{fig:sp500-length} shows that the BtLasso and BtLasso+OLS produce confidence intervals with lengths nearly zero, while the confidence interval lengths of the other four methods are larger than zero and roughly comparable.

\section{Conclusion}

Performing statistical inference for parameters in high-dimensional VAR models is a challenging but important problem.
We propose to use the de-biased Lasso (LDPE), residual bootstrap de-biased Lasso (BtLDPE) and multiplier wild bootstrap de-biased Lasso (MultiBtLDPE) to construct confidence intervals for the individual parameter of the transition matrix.
Unlike the fixed design case in linear regression models, the design matrix in VAR models is random with complex dependence structure, which makes theoretical analysis challenging.
Based on the convergence rates of the Lasso and the nodewise Lasso estimators, we obtain the asymptotic unbiasedness of the de-biased Lasso estimator. Combined with the martingale central limit theorem, we obtain its asymptotic normality.
For the two bootstrap de-biased Lasso methods, the analysis is conditional on the original data and we derive their asymptotic properties based on the randomness coming from the bootstrap sampling. We demonstrate the validity of statistical inference for the parameters of  high-dimensional sparse VAR models using these methods.
Furthermore, we propose feasible and parallelizable algorithms to implement our methods. 
More specifically, we apply the de-biased Lasso, residual bootstrap de-biased Lasso and multiplier wild bootstrap de-biased Lasso to each of the $p$ equations of VAR models separately, which can be ran in parallel. 
More importantly, the $p$ equations share the same design matrix, so we need only to compute the nodewise Lasso once, which is the main computational burden of the de-biased Lasso and bootstrap de-biased Lasso. 
The proposed methods have significant computational advantages, especially when $p$ is large.

We conduct comprehensive simulation studies to compare our methods with the bootstrap Lasso, bootstrap Lasso+OLS and another de-biased Lasso method proposed by \cite{Javanmard:2014}. We find that the LDPE can always give the honest coverage probabilities, and when sample size is large, the BtLDPE and MultiBtLDPE can also give the honest coverage probabilities but with shorter confidence interval lengths. Therefore, when the sample size is small, we recommend the LDPE for its reliability, and when the sample size is large, we recommend the BtLDPE and MultiBtLDPE for their reliability and power.
Lastly, we apply our methods to analyze the S\&P 500 constituent stock prices data set and obtain reasonable confidence intervals.

In our theoretical study, we assume the homoscedastic Gaussian errors. However, we find in simulations that our methods are robust to heteroscedastic and/or non-Gaussian errors. It is interesting and worthy of further investigation to obtain the asymptotic distributions of the proposed methods for heteroscedastic and/or non-Gaussian errors.  The main technical difficulty is to establish convergence rates of the Lasso and martingale central limit theorem for this type of errors.

This article focuses on statistical inference for individual parameter of the transition matrix in high-dimensional sparse VAR models. It is interesting to extend the methods for simultaneous confidence intervals and multiple hypothesis testing. For this purpose, we can use the Bonferroni correction or Westfall-Young procedure \citep{Westfall:1993uu}. We leave the corresponding theoretical investigation to future work.

\section*{Supplementary Material}

The document provides the detailed proofs of the theoretical results, as well as additional simulation results for different distributions of errors.

\section*{Acknowledgments}

The authors thank Dr. Lixiang Zhang for his suggestions that have helped clarify the text.

\section*{Funding}

Dr. Hanzhong Liu acknowledges the financial support from the National Natural Science Foundation of China (grant nos. 11701316).

\bibliographystyle{natbib.bst}
\bibliography{hdciVAR.bib}

\appendix

\begin{center}
{\large\bf Supplementary Material for ``Confidence Intervals for Parameters in High-dimensional Sparse Vector Autoregression"}
\end{center}

The document provides the detailed proofs of the theoretical results in the main text, as well as additional simulation results for different distributions of errors.

\section{Proofs of the theoretical results}

Our proofs require several theoretical results from \cite{Basu:2015ho}. 
Firstly, we bound the extreme eigenvalue of $\Sigma$, which can be obtained directly from  Proposition 2.3 in \cite{Basu:2015ho}.
\begin{proposition}\label{prop:prop2.3}
	For the stable VAR model in (\ref{eq:VAR}), we have
	$$
	1/\Lambda_{\rm min}(\Sigma)=O(1), \quad \Lambda_{\rm max}(\Sigma)=O(1).
	$$
\end{proposition}

Secondly, the following concentration inequalities are obtained from Proposition 2.4 in \cite{Basu:2015ho}.
\begin{proposition}\label{prop:prop2.4}
	For the stable VAR model in (\ref{eq:VAR}), \\(a) There exists a constant $c_1>0$ such that for any vector $v\in \mathbb{R}^p$ with $\|v\|_2 \leq 1$, for any $\eta \geq 0$, we have
	\begin{equation}
		P[v^{\T}(\mathbf { X } ^ {\T } \mathbf { X } / n-\Sigma)v>Q_1\eta]\leq 2 \exp \left[-c_1 n \min \left\{\eta^{2}, \eta\right\}\right],\label{eq:concen1}
	\end{equation}
	where $Q_1$ is a constant depends only on the dependence structure of the VAR model.\\
	(b) For $(p-1)\times p$ matrix $W$ and $p$-dimensional vector $w$, if ${\rm Cov}(W\mathbf{y}_{t},w^\T\mathbf{y}_{t})=0$ for every $t\in \mathbb{Z}$, then there exist constants $c_2$ such that for any vector $u\in \mathbb{R}^{p-1}$ with $\|u\|_2 \leq 1$, for any $\eta \geq 0$, we have
	\begin{equation}
		P[u^{\T}(W\mathbf { X } ^ {\T } \mathbf { X }w / n)>Q_2\eta]\leq 6 \exp \left[-c_2 n \min \left\{\eta^{2}, \eta\right\}\right],\label{eq:concen3}
	\end{equation}
	where $Q_2$ is a constant depends only on the dependence structure of the VAR model.
\end{proposition}

Thirdly, the following Proposition \ref{prop:4.2&4.3} are from Proposition 4.2 and 4.3 in \cite{Basu:2015ho}. The first inequality is related to the restricted eigenvalue (RE) condition for the design matrix $\mathbf{X}$ and we modify the original general conclusion to fit our proof related to the nodewise Lasso.
The second is the deviation bound for $\mathbf{X}^{\T}\varepsilon_i/n$. 
\begin{proposition}\label{prop:4.2&4.3}
(a) Under Assumptions \ref{ap:s0} and \ref{ap:sj}, we have, in probability,
\begin{equation}
	\theta^{\prime}  (\mathbf{X}^\T \mathbf{X}/n) \theta \geq \alpha\|\theta\|_2^{2}-\tau\|\theta\|_{1}^{2}, \quad \forall \theta \in \mathbb{R}^{p},
	\label{eq:ndl-RE}
\end{equation}
where $\alpha>0$ is a constant depends only on the dependence structure of the VAR model, $\tau$ satisfies $\tau s_{i}=O(1)$ for any $i$ and $\tau q_{j}=O(1)$ for any $j$.
\\(b) Under Assumption \ref{ap:s0}, for $i=1,...,p$, we have, in probability,
\begin{equation}
	\|\mathbf{X}^{\T}\varepsilon_i/n\|_\infty\leq C_1\sqrt{\log p/n}, \label{eq:L2}
\end{equation}
where $C_1>0$ is a constant depends only on the dependence structure of the VAR model.
\end{proposition}

Lastly, we need Lemma F.2 in the supplementary material of \cite{Basu:2015ho}, which used discretization to expand the bound of a single vector in (\ref{eq:concen1}) to a set of sparse vectors.
\begin{lemma}\label{lm:discre}
	Consider a symmetric matrix $D_{p \times p}$. If, for any vector $v\in \mathbb{R}^p$ with $\|v\|_2 \leq 1$, and any $\eta \geq 0$,
	$$
		P[v^{\T}Dv>C\eta]\leq 2 \exp \left[-c n \min \left\{\eta^{2}, \eta\right\}\right],
	$$
	then, for any integer $s\geq 1$, we have
	$$
	P\left[\sup _{\|v\|_0\le s} \sup _{\|v\|_2\leq 1}\left|v^{\prime} D v\right|>C \eta\right] \leq 2 \exp \left[-c n \min \left\{\eta, \eta^{2}\right\}+s \min \{\log p, \log (21 e p / s)\}\right].
	$$
\end{lemma}

\subsection{Proofs of the theoretical results in Section \ref{sec:dbl-th}}

We first show that the following sparse Riesz condition \citep{Zhang:2008ga} holds with probability converging to 1:
\begin{equation}\label{eq:Riesz}
	c_{*} \leq \min _{\|v\|_0\le s_R} \min _{\|v\|_2=1}\left\|\mathbf{X} v\right\|^{2}_2 / n \leq \max _{\|v\|_0\le s_R} \max _{\|v\|_2=1}\left\|\mathbf{X} v\right\|^{2}_2 / n \leq c^{*}.
\end{equation}

\begin{proof}[\textbf{Proof of Proposition \ref{prop:Riesz}}]
	By (\ref{eq:concen1}) in Proposition \ref{prop:prop2.4} and Lemma \ref{lm:discre}, we have
	$$
	\begin{aligned}
		P\left[\sup _{\|v\|_0\le s_R} \sup _{\|v\|_2\leq 1}\left|v^{\prime} (\mathbf { X } ^ {\T } \mathbf { X } / n-\Sigma) v\right| >Q_1\eta\right]
		\leq 2 \exp [&-c n \min \{\eta, \eta^{2}\}\\&+s_R \min \{\log p, \log (21 e p / s_R)\}].
	\end{aligned}
	$$
	Thus, we have
	$$
	\sup _{\|v\|_0\le s_R} \sup _{\|v\|_2\leq 1}\left|v^{\prime} (\mathbf { X } ^ {\T } \mathbf { X } / n-\Sigma) v\right| =O_p(\sqrt{ s_R \log p/n}) =O_p(\sqrt{ s_{i} \log p/n}) = o_p(1),
	$$
	where the last equality is due to the sparsity Assumption \ref{ap:s0} ($s_{i} \log p/ \sqrt{n} \rightarrow 0$).
	By Proposition \ref{prop:prop2.3}, we have
	$$
	\sup _{\|v\|_0\le s_R} \sup _{\|v\|_2\leq 1}\left|v^{\prime} \Sigma v\right| =O_p(1), \quad 1 \Big /  \min _{\|v\|_0\le s_R} \min _{\|v\|_2\leq 1}\left|v^{\prime} \Sigma v\right| =O_p(1).
	$$
	The result follows from triangle inequality.
\end{proof}

Now, we can prove Theorem \ref{thm:Lasso}. Since the statement (a) have been obtained by \cite{Basu:2015ho} (see Propositions 4.1), we  need only to prove the statement (b).

\begin{proof}[\textbf{Proof of Theorem \ref{thm:Lasso} (b)}]
	Firstly, we introduce some notations. Note that we omit the subscript $i$ for simplicity.
		We define the sets as follows:
	\begin{gather*}
	A_{tp} := \{j: a_{ij}\neq 0\}\cap\{j: \hat a^{\rm Lasso}_{ij}\neq 0\},\\
	A_{fp} := \{j: a_{ij}= 0\}\cap\{j: \hat a^{\rm Lasso}_{ij}\neq 0\},\\
	A_{fn} := \{j: a_{ij}\neq 0\}\cap\{j: \hat a^{\rm Lasso}_{ij}= 0\},\\
	A_{tn} := \{j: a_{ij}= 0\}\cap\{j: \hat a^{\rm Lasso}_{ij}= 0\},\\
	A_1 := A_{tp}\cup A_{fp}\cup A_{fn},  \quad
	A_2 := A_{tp}\cup A_{fn},  \quad
	A_3 := A_{fp}, 	
	\end{gather*}
	where the subscripts represent true positive, false positive, false negative and true negative. We define an important quantity as
	$$
	\hat s_{+} := |A_1|.
	$$
	By these definitions, we have 
	\begin{equation}
		s_{i} = |A_2|\leq |A_1|=\hat s_{+}, \quad 
		\hat s_{i}=|A_{tp}\cup A_{fp}| \leq |A_1|=\hat s_{+}, \quad
		|A_3|=\hat s_{+} - s_{i}.
		\label{eq:s_hat_and_q}
	\end{equation}
	Since $A_2\subseteq A_1$ and $A_3\subseteq A_1$, for $k=2,3$, let $\mathbf{Q}_{k1}$ be an $|A_k| \times |A_1|$ transformation matrix with elements 0 and 1. $\mathbf{Q}_{k1}$ selects variables in $A_k$ from $A_1$, such that $\mathbf{Q}_{k1}\mathbf{X}_{A_1}^\T=\mathbf{X}_{A_k}^\T$.
	Recall that 
	$$
	\hat a^{\rm Lasso}_{i} := \mathop{{\rm argmin}}\limits_{\alpha\in \mathbb{R}^p}\{ \|Y_i - \mathbf { X } \alpha\|_2^2/n + 2\lambda\|\alpha\|_1\},
	$$
	where $\hat a^{\rm Lasso}_{i}=(\hat a^{\rm Lasso}_{i1},...,\hat a^{\rm Lasso}_{ip})$.
	By KKT condition, we have
	\begin{equation}
		\left\{\begin{array}{ll}
	{ X_j^{\T}(Y_i-\mathbf{X}\hat a^{\rm Lasso}_{i})/n =\lambda\operatorname { sign } ( \hat { a } ^{\rm Lasso}_ { ij } ),} & {\hat { a }^{\rm Lasso} _ {i j } \neq 0} \\
	{|X_j^{\T}(Y_i-\mathbf{X}\hat a^{\rm Lasso}_{i})/n|\leq\lambda}, & {\hat { a }^{\rm Lasso} _ {i j } = 0}
	\end{array}\right..
	\label{eq:KKT}
	\end{equation}
	For $k=2,3$, we denote
	\begin{equation}
		\xi_k := \mathbf{X}_{A_k}^{\T}(Y_i-\mathbf{X}\hat a^{\rm Lasso}_{i})/(n\lambda),
	\end{equation}
	and
	\begin{equation}
		\mathbf{v}_k := \lambda\sqrt{n}\,\hat\Sigma_{A_1,A_1}^{-1/2}\mathbf{Q}_{k1}^{\T}\xi_k,
	\end{equation}
	where $\hat\Sigma_{A_1,A_1}:=\mathbf{X}_{A_1}^{\T}\mathbf{X}_{A_1}/n$. By (\ref{eq:KKT}), we have
	\begin{equation}
		\|\xi_2\|^2_2\leq |A_2|, \quad \|\xi_3\|^2_2 = |A_3| = \hat s_{+} - s_{i}.
		\label{eq:xi_norm}
	\end{equation}
	
	Next, we intend to proving that $\hat s_{+}=O_p(s_{i})$, which implies $\hat s_{i}=O_p(s_{i})$ by (\ref{eq:s_hat_and_q}). Our proof is divided into three steps. 
	In steps 1 and 2, taking $\mathbf{v}_3$ as a bridge, we prove that $\hat s_{+}\leq s_R$ implies $\hat s_{+} \leq (2 + 2c^*/c_*)s_{i}$. Specifically, we give a lower bound and an upper bound of $\|\mathbf{v}_3\|_2^2$ in steps 1 and 2 respectively.
	In step 3, using the results in steps 1 and 2, we proves the desired result by contradiction.
	
	\textit{Step} 1. Assuming that $\{\hat s_{+}\leq s_R\}$, by (\ref{eq:Riesz}) and (\ref{eq:xi_norm}), we have
	\begin{equation}
	\begin{aligned}
		\|\mathbf{v}_3\|_2^2= & \,\lambda^2 n\, \xi_3^{\T}\mathbf{Q}_{31}\hat\Sigma_{A_1,A_1}^{-1}\mathbf{Q}_{31}^{\T}\xi_3 \\
		\geq & \,\lambda^2 n\, \|\mathbf{Q}_{31}^{\T}\xi_3\|^2_2/c^*\\
		=& \,\lambda^2 n\, (\hat s_{+} - s_{i})/c^*.
	\end{aligned}
	\label{eq:v3low}
	\end{equation}
	
	\textit{Step} 2. Assuming $\{\hat s_{+}\leq s_R\}$, since $a_{i,A_{tn}}=\hat a^{\rm Lasso}_{i,A_{tn}}=\mathbf 0$ and $A_2 \cup A_3 = A_1$, we have
	\begin{equation}
	\begin{aligned}
		n\lambda(\mathbf{Q}_{21}^{\T}\xi_2 + \mathbf{Q}_{31}^{\T}\xi_3)
		= &\, \mathbf{X}_{A_1}^{\T}(Y_i-\mathbf{X}\hat a^{\rm Lasso}_{i}) \\
		= &\, \mathbf{X}_{A_1}^{\T}(Y_i-\mathbf{X}_{A_1}\hat a^{\rm Lasso}_{i,A_1}) \\
		= &\, \mathbf{X}_{A_1}^{\T} \mathbf{X}_{A_1} a_{i,A_1}
			+ \mathbf{X}_{A_1}^{\T} \mathbf{X}_{A_{tn}} a_{i,A_{tn}}
			+ \mathbf{X}_{A_1}^{\T} \varepsilon_i
			- \mathbf{X}_{A_1}^{\T} \mathbf{X}_{A_1}\hat a^{\rm Lasso}_{i,A_1}\\
		= &\, n\hat\Sigma_{A_1,A_1}a_{i,A_1} 
			+ \mathbf{X}_{A_1}^{\T} \varepsilon_i
			- n\hat\Sigma_{A_1,A_1}\hat a^{\rm Lasso}_{i,A_1}.
	\end{aligned}
	\label{eq:xi2+3}
	\end{equation}
	Then we have
	\begin{equation}
	\begin{aligned}
		\mathbf{v}_3^{\T}(\mathbf{v}_2+\mathbf{v}_3) 
		= & \,\lambda^2 n\, \xi_3^{\T}\mathbf{Q}_{31}\hat\Sigma_{A_1,A_1}^{-1}(\mathbf{Q}_{21}^{\T}\xi_2 + \mathbf{Q}_{31}^{\T}\xi_3) \\
		= & \, \lambda \, \xi_3^{\T}\mathbf{Q}_{31}\hat\Sigma_{A_1,A_1}^{-1}\mathbf{X}_{A_1}^{\T} \varepsilon_i
		+ \lambda n \, \xi_3^{\T}\mathbf{Q}_{31}(a_{i,A_1} - \hat a^{\rm Lasso}_{i,A_1}) \\
		\leq & \, \lambda \, \xi_3^{\T}\mathbf{Q}_{31}\hat\Sigma_{A_1,A_1}^{-1}\mathbf{X}_{A_1}^{\T} \varepsilon_i\\
		\leq & \, \lambda n \, \|\xi_3^{\T}\mathbf{Q}_{31}\hat\Sigma_{A_1,A_1}^{-1}\|_1 \|\mathbf{X}_{A_1}^{\T} \varepsilon_i/n\|_{\infty},
	\end{aligned}
	\label{eq:v3(v2+v3)}
	\end{equation}
	where the first inequality is due to
	$$
	\mathbf{Q}_{31}(a_{i,A_1} - \hat a^{\rm Lasso}_{i,A_1}) = a_{i,A_3} - \hat a^{\rm Lasso}_{i,A_3}=- \hat a^{\rm Lasso}_{i,A_3},
	$$
	and
	$$
	\xi_3^{\T}\hat a^{\rm Lasso}_{i,A_3}\geq 0,
	$$
	by the KKT condition (\ref{eq:KKT}), and the second inequality is due to H\"older inequality.
	Furthermore, we have
	\begin{equation}
		\begin{aligned}
			\|\xi_3^{\T}\mathbf{Q}_{31}\hat\Sigma_{A_1,A_1}^{-1}\|_1
			\leq & \sqrt{\hat s_{+}}\,\|\xi_3^{\T}\mathbf{Q}_{31}\hat\Sigma_{A_1,A_1}^{-1}\|_2 \\
			\leq & \sqrt{\hat s_{+}}\, \|\xi_3^{\T}\mathbf{Q}_{31}\|_2/c_* \\
			\leq & \hat s_{+} /c_*,
		\end{aligned}
		\label{eq:xi_Q_sigma}
	\end{equation}
	where the second inequality is due to the sparse Riesz condition (\ref{eq:Riesz}) and the last inequality is because $\|\xi_3^{\T}\mathbf{Q}_{31}\|_2 = \|\xi_3 \|_2 = \sqrt{\hat s_{+} - s_{i}}$. By sparse Riesz condition (\ref{eq:Riesz}) and (\ref{eq:xi_norm}), we also have
	\begin{equation}
	\begin{aligned}
		\|\mathbf{v}_2\|_2^2= & \,\lambda^2 n\, \xi_2^{\T}\mathbf{Q}_{21}\hat\Sigma_{A_1,A_1}^{-1}\mathbf{Q}_{21}^{\T}\xi_2 \\
		\leq & \,\lambda^2 n\, \|\mathbf{Q}_{21}^{\T}\xi_2\|^2_2/c_*\\
		\leq & \,\lambda^2 n\, s_{i}/c_*.
	\end{aligned}
	\label{eq:v2_norm}
	\end{equation}
	By (\ref{eq:v3(v2+v3)}), triangle inequality, H\"older inequality, (\ref{eq:xi_Q_sigma}) and (\ref{eq:v2_norm}), we have
	\begin{equation}
		\begin{aligned}
			\|\mathbf{v}_3\|_2^2 
			\leq & \, \lambda n \, \|\xi_3^{\T}\mathbf{Q}_{31}\hat\Sigma_{A_1,A_1}^{-1}\|_1 \|\mathbf{X}_{A_1}^{\T} \varepsilon_i/n\|_{\infty} + \|\mathbf{v}_2\|_2\|\mathbf{v}_3\|_2 \\
			\leq & \, (\lambda n  \hat s_{+} /c_* ) \|\mathbf{X}_{A_1}^{\T} \varepsilon_i/n\|_{\infty} + (\lambda^2 n s_{i}/c_*)^{1/2}\|\mathbf{v}_3\|_2. \\
		\end{aligned}
		\label{eq:v322}
	\end{equation}
	Now we need to bound the term $\|\mathbf{X}_{A_1}^{\T} \varepsilon_i/n\|_{\infty}$. By (\ref{eq:L2}), we have, in probability,
	\begin{equation}
		\|\mathbf{X}_{A_1}^{\T} \varepsilon_i/n\|_{\infty}\leq\|\mathbf{X}^{\T} \varepsilon_i/n\|_{\infty}\leq C_1 \sqrt{\log p/n}.
		\label{eq:Xe_inf}
	\end{equation}
	Since $x^2\leq c+2bx$ implies $x^2\leq (b+\sqrt{b^2+c})^2\leq 2c+4b^2 $ for $x=\|\mathbf{v}_3\|_2$, by (\ref{eq:v322}) and  (\ref{eq:Xe_inf}), we have, in probability,
	\begin{equation}
		\|\mathbf{v}_3\|_2^2 \leq 2 (\lambda n  \hat s_{+} /c_* ) C_1 \sqrt{\log p/n} + \lambda^2 n s_{i}/c_*.
	\end{equation}
	Combining with the lower bound (\ref{eq:v3low}), we have
	$$
	\lambda^2 n\, (\hat s_{+} - s_{i})/c^*\leq 2 (\lambda n  \hat s_{+} /c_* ) C_1 \sqrt{\log p/n} + \lambda^2 n s_{i}/c_*.
	$$
	Moving the terms related to $\hat s_{+}$ and $s_{i}$ to each side of the inequality, we have
	$$
	\left(\lambda -\frac{2C_1c^*}{c_*}\sqrt{\frac{\log p}{n}}\right)\hat s_{+}\leq \left(1+\frac{c^*}{c_*}\right)\lambda s_{i}.
	$$
	Recall the condition in Theorem \ref{thm:Lasso},
	$$
	\lambda \geq \frac{4C_1c^*}{c_*}\sqrt{\frac{\log p}{n}},
	$$
	we have
	\begin{equation}
		\frac12 \lambda\hat s_{+} \leq
		\left(\lambda -\frac{2C_1c^*}{c_*}\sqrt{\frac{\log p}{n}}\right)\hat s_{+}
		\leq \left(1+\frac{c^*}{c_*}\right)\lambda s_{i},
		\nonumber
	\end{equation}
	then
	\begin{equation}
		\hat s_{+} \leq (2 + 2c^*/c_*)s_{i}.
		\nonumber
	\end{equation}
	Thus, we obtain that 
	\begin{equation}
		\hat s_{+}\leq s_R\quad \Rightarrow \quad \hat s_{+} \leq (2 + 2c^*/c_*)s_{i}.
		\label{eq:qhat_bound}
	\end{equation}

	\textit{Step} 3. 
	Note that $\hat s_{+}=\hat s_{+}(\lambda)$ is a function of $\lambda$. 
	We denote
	$$
	\hat s_{+,{\rm min}}:=\min_{\lambda\geq 4C_1c^*/c_*\sqrt{\log p/n}}\hat s_{+}(\lambda),\quad \hat s_{+,{\rm max}}:=\max_{\lambda\geq 4C_1c^*/c_*\sqrt{\log p/n}}\hat s_{+}(\lambda).
	$$
	Because of the continuity of the Lasso path, we could choose the variable one-at-a-time. Therefore, when $\lambda\geq 4C_1c^*/c_*\sqrt{\log p/n}$, $\hat s_{+}(\lambda)$ could take every integer from $\hat s_{+,{\rm min}}$ to $\hat s_{+,{\rm max}}$.
	
	When $\lambda=\infty$, we have $A_1=A_2$ and thus 
	\begin{equation}
		\hat s_{+,{\rm min}}=s_i\leq (2 + 2c^*/c_*)s_{i}.
		\nonumber
	\end{equation}
	If $\hat s_{+,{\rm max}}>s_R$, by the definition of $s_R$
	in Proposition \ref{prop:Riesz}, we have 
	\begin{equation}
		\hat s_{+,{\rm max}}>s_R>(2 + 2c^*/c_*)s_{i}+1.
		\nonumber
	\end{equation}
	Then, there exists a $\lambda \geq 4C_1c^*/c_*\sqrt{\log p/n}$, such that $\hat s_{+}=\lfloor(2 + 2c^*/c_*)s_{i}+1\rfloor$, which contradicts (\ref{eq:qhat_bound}) since $\hat s_{+}\leq s_R$ and $\hat s_{+} > (2 + 2c^*/c_*)s_{i}
$. By contradiction, $\hat s_{+,{\rm max}}\leq s_R$ and our result follows.
\end{proof}

Next, we prove the estimation and the prediction error bounds for the nodewise Lasso estimator defined in (\ref{eq:nodelasso}). These bounds are useful in the proof of Theorem \ref{thm:dbl}. Recall that these bounds are
\begin{gather}
	\| \hat { \gamma } _ { j } - \gamma _ { j } \| _ { 1 } = O_p \left(q _ { j }\sqrt{\log p/n}\right),\label{eq:ndl-1norm} \\
	\| \mathbf { X } _ { - j } \left( \hat { \gamma } _ { j } - \gamma _ { j } \right) \| _ { 2 } ^ { 2 } / n = O_p \left(q _ { j } { \log p/n}\right),\label{eq:ndl-pre-err}
\end{gather}
where $j=1,...,p$.

\begin{proof}[\textbf{Proof of Proposition \ref{prop:nodelasso}}]
For $j=1,...,p$, since $\hat { \gamma } _j$ is the minimizer:
\begin{equation}
\hat{\gamma}_j := \mathop{{\rm argmin}}\limits_{\gamma\in \mathbb{R}^{p-1}}\{ ||X_j-\mathbf{X}_{-j}\gamma||_2^2/n + 2\lambda_j||\gamma||_1\},
\nonumber
\end{equation}
we obtain the basic inequality
\begin{equation}
	 \|X_j-\mathbf{X}_{-j}\hat { \gamma } _j\|_2^2/n + 2\lambda_j\|\hat { \gamma } _j\|_1\leq
	 \|X_j-\mathbf{X}_{-j}\gamma_j\|_2^2/n + 2\lambda_j\|\gamma_j\|_1.
	 \nonumber
\end{equation}
We denote $\delta:= \hat { \gamma } _ { j } - \gamma _ { j }$. By simple algebra, we have 
\begin{equation}
	\| \mathbf { X } _ { - j } \delta \| _ { 2 } ^ { 2 } / n \leq
	2\tilde\varepsilon^{\T}\mathbf { X } _ { - j }\delta/n + 2\lambda_j(\|\gamma_j\|_1-\|\hat { \gamma } _j\|_1),
	\label{eq:ndl-basic}
\end{equation}
where $\tilde\varepsilon:=X_j-\mathbf{X}_{-j}\gamma_j$. Since ${\rm Cov}(\mathbf { X } _ { - j },\tilde\varepsilon)=0$, then by (\ref{eq:concen3}), we have
\begin{equation}
	P\left(\|\tilde\varepsilon^{\T}\mathbf { X } _ { - j }\|_\infty> 
		Q_2 \eta \right)
		\leq 6p \exp \left[ - c_3 n \min \left\{ \eta , \eta ^ { 2 } \right\} \right].
		\nonumber
\end{equation}
With $\eta=c_0\sqrt{\log p/n}$ and suitable chosen $\lambda_j\asymp\sqrt{\log p/n}$, we have, in probability,
\begin{equation}
	\|\tilde\varepsilon^{\T}\mathbf { X } _ { - j }\|_\infty\leq Q_2c_0\sqrt{\log p/n}\leq\lambda_j/2.
	\nonumber
\end{equation}
Furthermore, by H\"older inequality, we have
\begin{equation}
	\tilde\varepsilon^{\T}\mathbf { X } _ { - j }\delta/n\leq \|\tilde\varepsilon^{\T}\mathbf { X } _ { - j }\|_\infty\| \delta\|_1\leq\lambda_j\| \delta\|_1/2.
	\nonumber
\end{equation}
We denote the support set of $\gamma_j$ by $K$. By triangle inequality and $\gamma_{jK^C}=0$, we have
\begin{equation}
	\|\gamma_j\|_1-\|\hat { \gamma } _j\|_1=\|\gamma_{jK}\|_1-\|\hat { \gamma } _{jK}\|_1+\|\gamma_{jK^C}\|_1-\|\hat { \gamma } _{jK^C}\|_1\leq \|\delta_{K}\|_1-\|\delta_{K^C}\|_1.
	\nonumber
\end{equation}
Therefore, (\ref{eq:ndl-basic}) becomes
\begin{equation}
	\begin{aligned}
		0\leq\| \mathbf { X } _ { - j } \delta \| _ { 2 } ^ { 2 } / n &\leq \lambda_j\| \delta\|_1+ 2\lambda_j(\|\delta_{K}\|_1-\|\delta_{K^C}\|_1)\\
		&= \lambda_j(\|\delta_{K}\|_1+\|\delta_{K^C}\|_1)+ 2\lambda_j(\|\delta_{K}\|_1-\|\delta_{K^C}\|_1)\\
		&=3\lambda_j\|\delta_{K}\|_1-\lambda_j\|\delta_{K^C}\|_1\\
		&\leq 3\lambda_j\|\delta\|_1.
	\end{aligned}
	\label{eq:ndl-1}
\end{equation}
Inequality (\ref{eq:ndl-1}) implies $\|\delta_{K^C}\|_1\leq3\|\delta_{K}\|_1$ so that 
\begin{equation}
	\|\delta\|_1\leq 4\|\delta_{K}\|_1\leq4\sqrt{q_j}\|\delta\|_2.
	\label{eq:ndl-2}
\end{equation}
By (\ref{eq:ndl-RE}) and (\ref{eq:ndl-2}), we have
\begin{equation}
	\| \mathbf { X } _ { - j } \delta \| _ { 2 } ^ { 2 } / n\geq \alpha\|\delta \|_2^2-\tau \|\delta \|_1^2\geq
	\left(\frac{\alpha}{4q_j}-\tau \right)\|\delta\|_1^2.
	\label{eq:ndl-3}
\end{equation}
By (\ref{eq:ndl-1}), (\ref{eq:ndl-3}), $\tau q_j=O(1)$ and Assumption \ref{ap:lambda} ($\lambda_j\asymp \sqrt{\log(p)/n}$), we obtain (\ref{eq:ndl-1norm}). Then by (\ref{eq:ndl-1}), (\ref{eq:ndl-1norm}) and Assumption \ref{ap:lambda}, we obtain (\ref{eq:ndl-pre-err}).
\end{proof}

To prove Theroem~\ref{thm:dbl}, compared with the proof of the validity of the de-biased Lasso in high-dimensional sparse linear regression models \citep{vandeGeer:2014}, our proof is challenge because of dependence structure of data generating process. We denote the nodewise Lasso residuals by
$$
\hat Z_{j}=(\hat{Z}_{1j},...,\hat{Z}_{nj}):=X_j-\mathbf{X}_{-j}\hat\gamma_j,
$$
which is an estimator of
$$
Z_{j}=( {Z}_{1j},..., {Z}_{nj}):=X_j-\mathbf{X}_{-j} \gamma_j.
$$
We also denote 
$$
\tau _ { j } ^ { 2 } := E(\|Z_j\|_2^2/n).
$$
Unlike the proof for high-dimensional sparse linear regression models, we should distinguish $\hat{Z}_{j}$ and ${Z}_{j}$ carefully because of the correlation between $\mathbf{X}_{-j}$ and $\hat\gamma_j$.

\begin{proof}[\textbf{Proof of Theorem \ref{thm:dbl}}]
Recall that
\begin{equation}\label{eq:KKT-3recall}
	\hat { a }_i - a_i   = \hat\Theta\mathbf { X } ^ {\rm  T } \varepsilon_i / n + (I-\hat\Theta\hat { \Sigma })(\hat { a }^{\rm Lasso}_i - a_i) .
\end{equation}
We prove the theorem in three steps. Step 1 proves that the second term of the right-hand of (\ref{eq:KKT-3recall}) is asymptotically negligible. Step 2 proves that the first term is asymptotically normal. Step 3 proves that our variance estimator, $\hat { \sigma } _ { i } ^ { 2 } =   \| \hat \varepsilon_i \| _ { 2 } ^ { 2 }/(n - \hat { s }_{i0} )$, is consistent.

\textit{Step} 1. Recall that
\begin{gather*}
\hat{C}:=   \left[ \begin{matrix}
   1 & -\hat{\gamma}_{12} &  \cdots & -\hat{\gamma}_{1p} \\
   -\hat{\gamma}_{21} & 1 &  \cdots & -\hat{\gamma}_{2p} \\
    \vdots &  \vdots & \ddots & \vdots \\
   -\hat{\gamma}_{p1} & -\hat{\gamma}_{p2} & \cdots &1
  \end{matrix} \right], \\
\hat{\tau}_{j}^{2}:=\left\|X_{j}-\mathbf{X}_{-j} \hat{\gamma}_{j}\right\|_{2}^{2} / n+\lambda_{j}\left\|\hat{\gamma}_{j}\right\|_{1}, \quad \hat { T } ^ { 2 } : = \operatorname { diag } ( \hat { \tau } _ { 1 } ^ { 2 } , \ldots , \hat { \tau } _ { p } ^ { 2 } ), \quad \hat { \Theta }  : = \hat { T } ^ { - 2 } \hat { C }.
\end{gather*}
We denote the $j$th row of $\hat\Theta$ by $\hat\Theta_j$ and obtain
\begin{equation}\label{Theta_j}
	\mathbf{X}\hat\Theta_j^\T/n=( X_j -\mathbf { X }_{-j} \hat{\gamma}_{j} ) / (n \hat{\tau}_{j}^{2}).
\end{equation}
By the KKT conditions for the nodewise Lasso, we have
\begin{equation}\label{eq:KKT-nodelasso-1}
  -\mathbf { X }_{-j} ^ {\T } ( X_j -\mathbf { X }_{-j} \hat{\gamma}_{j} ) / n +\lambda_j \hat { \kappa }^j=0,
\end{equation}
where $\hat { \kappa }$ is the sub-gradient of $\ell_1$ norm and satisfies $\|\hat { \kappa }^j\|_\infty\leq 1$ and $\hat { \kappa }^j_k=\operatorname { sign } ( \hat{\gamma}_{jk} )$ if $\hat{\gamma}_{jk}  \neq 0$. Multiplying both hand sides of (\ref{eq:KKT-nodelasso-1}) by $\hat\gamma_j$, we obtain
\begin{equation}
	\lambda_j \|\hat{\gamma}_{j}\|=\hat\gamma_j^\T\mathbf { X }_{-j} ^ {\T } ( X_j -\mathbf { X }_{-j} \hat{\gamma}_{j} ) / n .
	\nonumber
\end{equation}
Substituting the above formula into the definition of $\hat{\tau}_{j}^{2}$, we have
\begin{equation*}
\hat{\tau}_{j}^{2}:=\left\|X_{j}-\mathbf{X}_{-j} \hat{\gamma}_{j}\right\|_{2}^{2} / n+\lambda_{j}\left\|\hat{\gamma}_{j}\right\|_{1}=X_j^\T( X_j -\mathbf { X }_{-j} \hat{\gamma}_{j} ) / n .
\end{equation*}
Combining with (\ref{Theta_j}), we have
\begin{equation}\label{eq:KKT-nodelasso-2}
	X_j^\T\mathbf{X}\hat\Theta_j^\T/n=X_j^\T( X_j -\mathbf { X }_{-j} \hat{\gamma}_{j} ) / (n \hat{\tau}_{j}^{2})=1.
\end{equation}
By (\ref{Theta_j}) and (\ref{eq:KKT-nodelasso-1}), we have
\begin{equation}\label{eq:KKT-nodelasso-3}
	\|\mathbf { X }_{-j}^\T\mathbf{X}\hat\Theta_j^\T/n\|_\infty=\|\mathbf { X }_{-j}^\T( X_j -\mathbf { X }_{-j} \hat{\gamma}_{j} )/(n \hat{\tau}_{j}^{2})\|_\infty\leq \lambda_j/ \hat{\tau}_{j}^{2}.
\end{equation}
By (\ref{eq:KKT-nodelasso-2}) and (\ref{eq:KKT-nodelasso-3}), we have
\begin{equation*}
	\| (I- \hat { \Theta } \hat { \Sigma }  ) \|_\infty=\max_j\| (e_j- \mathbf{X}^\T \mathbf{X}\hat\Theta_j^\T/n ) \|_\infty\leq \max _ { j } \lambda _ { j } / \hat { \tau } _ { j } ^ { 2 } .
\end{equation*}
By H{\"o}lder inequality, we have
\begin{equation}
	\begin{aligned} 
	\sqrt { n } \| (I- \hat { \Theta }  \hat { \Sigma }  ) ( \hat a^{\rm Lasso}_{i} -a_i ) \| _ { \infty } 
	& \leq \sqrt { n }\| (I- \hat { \Theta } \hat { \Sigma }  ) \| _ { \infty } \| \hat a^{\rm Lasso}_{i} -a_i \| _ { 1 } \\
	& \leq \sqrt { n } \left( \max _ { j } \lambda _ { j } / \hat { \tau } _ { j } ^ { 2 } \right) \| \hat a^{\rm Lasso}_{i} -a_i \| _ { 1 } .
	\end{aligned}
	\label{eq:Delta-1}
\end{equation}
By Theorem \ref{thm:Lasso}, we have 
\begin{equation}\label{eq:Delta-2}
	\| \hat a^{\rm Lasso}_{i} -a_i \| _ { 1 } = O_p (s_{i}\sqrt{\log p/n}).\nonumber
\end{equation}
Then, we intend to prove that for $j=1,...,p$,
\begin{equation}\label{eq:tau-const}
	1 / \hat { \tau } _ { j } ^ { 2 } = O_p( 1 ).
\end{equation}
Since the proofs are the same for different $j$, for simplicity, we only consider $j=1$.
By simple algebra, we have the following bound for $\tau_1^2$, which is the population level counterpart of $ \hat { \tau } _ { 1 } ^ { 2 } $,
\begin{equation}\label{eq:tau-0}
		1/\tau_1^2=\Theta_{11}=e_1^\T\Sigma^{-1} e_1\leq 1/\Lambda_{\rm min}(\Sigma), \quad
	\tau^2_1\leq E[X_{11}]^2=\Sigma_{1,1}=e_1^\T \Sigma e_1\leq \Lambda_{\rm max}(\Sigma),
\end{equation}
where $1/\Lambda_{\rm min}(\Sigma)=O(1)$ and $\Lambda_{\rm max}(\Sigma)=O(1)$ by Proposition \ref{prop:prop2.3}.
By the definition and Cauchy-Schwarz inequality, we have,
\begin{equation}\label{eq:tau-1}
\begin{aligned}
|\hat { \tau } _ { 1 } ^ { 2 } - { \tau } _ { 1 } ^ { 2 } |
&\leq |\| \hat{Z} _ { 1 } \| _ { 2 } ^ { 2 } / n- { \tau } _ { 1 } ^ { 2 }|  + \lambda _ { 1 } \| \hat { \gamma } _ { 1 } \| _ { 1 } \\
&\leq |\| {Z} _ { 1 } \| _ { 2 } ^ { 2 } / n- { \tau } _ { 1 } ^ { 2 }| +\| \hat{Z} _ { 1 }- {Z} _ { 1 } \| _ { 2 } ^ { 2 } / n+2\| {Z} _ { 1 } \| _ { 2 } \| \hat{Z} _ { 1 }- {Z} _ { 1} \| _ { 2 } /n+\lambda _ { 1 } \| \hat { \gamma } _ { 1 } \| _ { 1 }.
\end{aligned}
\end{equation}
Since $\| \hat{Z} _ { 1 }- {Z} _ { 1 } \| _ { 2 } ^ { 2 } / n=\| \mathbf { X } _ { - 1} \left( \hat { \gamma } _ { 1 } - \gamma _ { 1 } \right) \| _ { 2 } ^ { 2 } / n = O_p \left(q _ { 1 } { \log p/n}\right)$ has been bounded by (\ref{eq:ndl-pre-err}), we intend to work out $|\| {Z} _ { 1 } \| _ { 2 } ^ { 2 } / n- { \tau } _ { 1 } ^ { 2 }|$ and $\| \hat { \gamma } _ { 1 } \| _ { 1 }$.

Before that, we first bound the corresponding population level terms $\|\gamma_1\|_2$ and $\|\gamma_1\|_1$. Recall that $\gamma_1=\Sigma_{-1,-1}^{-1}\Sigma_{21}$. We partition the covariance matrix $\Sigma$ as
$$
	\Sigma=
	\left[\begin{matrix}
	 \Sigma_{1,1} & \Sigma_{1,-1}\\
	 \Sigma_{-1,1} & \Sigma_{-1,-1}
	\end{matrix}\right].
$$
We denote $\Sigma_{11\cdot2}:=\Sigma_{1,1} -\Sigma_{1,-1} \Sigma_{-1,-1}^{-1}\Sigma_{-1,1}=\Sigma_{1,1} -\gamma_1^{\T} \Sigma_{-1,-1}\gamma_1$.
Since $|\Sigma_{11\cdot2}\|\Sigma_{-1,-1}|=|\Sigma|>0$, $|\Sigma_{-1,-1}|>0$ and $\Sigma_{11\cdot2}$ is an one by one matrix (thus its determinant is itself), we have $\Sigma_{1,1} -\gamma_1^{\T} \Sigma_{-1,-1}\gamma_1>0$.
Therefore,
$$
	\Sigma_{1,1}\geq\gamma_1^{\T} \Sigma_{-1,-1}\gamma_1\geq\|\gamma_1\|_2^2\Lambda_{\rm min}(\Sigma_{-1,-1})\geq\|\gamma_1\|_2^2\Lambda_{\rm min}(\Sigma).
$$
We obtain
\begin{equation}
	\|\gamma_1\|_2\leq \sqrt {\Sigma_{1,1}/\Lambda_{\rm min}(\Sigma)}=O(1).
	\label{eq:gamma_norm2}
\end{equation}
By Cauchy-Schwarz inequality,
\begin{equation}
	\|\gamma_1\|_1\leq \sqrt{\|\gamma_1\|_0}\|\gamma_1\|_2\leq\sqrt q_1\|\gamma_1\|_2\leq \sqrt {q_1\Sigma_{1,1}/\Lambda_{\rm min}(\Sigma)}=O(\sqrt{q_1}).
	\label{eq:gamma_norm1}
\end{equation}

Then we study the term $\| {Z} _ { 1 } \| _ { 2 } ^ { 2 } / n$. With $\tilde v:=(1 , -\hat{\gamma}_{1,2} , ... , -\hat{\gamma}_{1,p})^{\T}/\sqrt{\|\gamma_1\|_2^2+1}$, we have
\begin{equation}\label{eq:Z-1}
\begin{aligned}
		|\|Z_{1}\|_2^2/n-\tau_1^2|&=|\|X_1-\mathbf{X}_{-1}\gamma_1\|_2^2/n-E(\|X_1-\mathbf{X}_{-1}\gamma_1\|_2^2/n)|\\
		&=(\|\gamma_1\|_2^2+1)\left|\tilde v^{\T}(\mathbf { X } ^ {\T } \mathbf { X } / n-\Sigma)\tilde v\right|\\
		&=O_p(\sqrt{\log p/n})\\
		&=o_p(1),
\end{aligned}
\end{equation}
where the first equality is due to the definitions of $\|Z_{1}\|_2^2$ and $\tau_1^2$, the second is due to the definition of $\tilde v$, the third is due to (\ref{eq:gamma_norm2}) and (\ref{eq:concen1}) with $\eta=c_0 \sqrt{\log p/n}$, and the fourth is due to Assumption \ref{ap:s0}.

Now we can bound $\| \hat { \gamma } _ { 1 } \| _ { 1 }$. We have
\begin{equation}\label{eq:gamma-hat-1}
	\begin{aligned}
		 \lambda_1\left\| \hat { \gamma } _ { 1 } \right\| _ { 1 } &\leq  \lambda_1\left\| \gamma _ { 1 } \right\| _ { 1 } + \lambda_1\left\| \hat { \gamma } _ { 1 } - \gamma _ { 1 } \right\| _ { 1 } \\
		&=  O_p \left(\sqrt{q _ { 1 }  \log p/n}\right) + O_p \left(q _ { 1 } { \log p/n}\right) 
,
	\end{aligned}
\end{equation}
where the first equality is due to triangle inequality and the second is due to (\ref{eq:gamma_norm1}), (\ref{eq:ndl-1norm}) and Assumption \ref{ap:lambda}.

We continue (\ref{eq:tau-1}) with the results (\ref{eq:ndl-pre-err}), (\ref{eq:Z-1}), (\ref{eq:tau-0}), (\ref{eq:gamma-hat-1}) and Assumption \ref{ap:sj},
\begin{equation}\label{eq:Z-2}
\begin{aligned}
|\hat { \tau } _ { 1 } ^ { 2 } - { \tau } _ { 1 } ^ { 2 } |
&\leq |\| {Z} _ { 1 } \| _ { 2 } ^ { 2 } / n- { \tau } _ { 1 } ^ { 2 }| +\| \hat{Z} _ { 1 }- {Z} _ { 1 } \| _ { 2 } ^ { 2 } / n+2\| {Z} _ { 1 } \| _ { 2 } \| \hat{Z} _ { 1 }- {Z} _ { 1} \| _ { 2 } /n+\lambda _ { 1 } \| \hat { \gamma } _ { 1 } \| _ { 1 }\\
&=O_p \left(q _ { 1 } { \log p/n}\right) + O_p\left(1\right)O_p \left(\sqrt{q _ { 1 }  \log p/n}\right) + O_p \left(\sqrt{q _ { 1 }  \log p/n}\right) + O_p \left(q _ { 1 } { \log p/n}\right) \\
&=o_p(1).	
\end{aligned}
\end{equation}
Combining with (\ref{eq:tau-0}), we obtain (\ref{eq:tau-const}). Also, we summarize some intermediate results we obtain above as follows,
\begin{gather}
	|\|Z_{j}\|_2^2/n-\tau_j^2|=o_p(1),\label{eq:Z-1-j}\\
	|\hat { \tau } _ { j } ^ { 2 } - { \tau } _ { j } ^ { 2 } |=o_p(1). \label{eq:Z-2-j}
\end{gather}

By (\ref{eq:Delta-1}), (\ref{eq:tau-const}) and Assumptions \ref{ap:lambda} and \ref{ap:s0}, we have
\begin{equation}
	\sqrt { n } \| (I- \hat { \Theta }  \hat { \Sigma }  ) ( \hat a^{\rm Lasso}_{i} -a_i ) \| _ { \infty }=o _ p(1).
	\label{eq:Delta-final}
\end{equation}

\textit{Step} 2. For $i,j=1,\dots,p$,  recall that $\varepsilon_i=(e_{i1},...,e_{in})^{\T}$ and our goal is to prove that
\begin{equation}
	\frac{1}{\sigma_i\| \hat{Z} _ { j }\|_2 }\sum_{t=1}^{n}\hat{Z}_{tj}e_{it} \stackrel { d } { \rightarrow } \mathcal { N } ( 0,1 ).
	\label{eq:asym-norm-final}
\end{equation} 
Firstly, we prove the asymptotic normality of the sum of martingale difference sequence $n^{-1/2}\sum_{t=1}^{n}Z_{tj}e_{it}$. Let $\mathcal{A}_t:=\sigma(\mathbf{e}_0,...,\mathbf{e}_t)$, then $Z_{tj}e_{it}$ is $\mathcal{A}_t$-measurable and $\mathcal{A}_{t-1}$ is contained in $\mathcal{A}_t$.
Since $Z_{tj}=X_{tj}-\sum_{l\neq j} \gamma_{jl}X_{tl}$ is $\mathcal{A}_{t-1}$-measurable, we have 
\begin{equation}\label{eq:mclt1}
		E(Z_{tj}e_{it}|\mathcal{A}_{t-1})=0.
\end{equation}
Let $V_n^2:=\sum_{t=1}^{n}E(n^{-1}Z^2_{tj}e^2_{1t}|\mathcal{A}_{t-1})=\sigma_i ^ { 2 }n^{-1}\sum_{t=1}^{n}Z^2_{tj}$ and $v^2_n:=E(V^2_n)=\sigma_i ^ { 2 }\tau^2_j$. By (\ref{eq:Z-1-j}), we have 
\begin{equation}\label{eq:mclt2}
		V_n^2v^{-2}_n\stackrel { P }\rightarrow 1.
\end{equation}
Let $\sigma^2_{\mathbf{X}}=\max_{1\leq t\leq n,1\leq j\leq p}{\rm Var}(X_{tj})$. By Propposition \ref{prop:prop2.3}, we have $\sigma^2_{\mathbf{X}}=O(1)$. Since the elements of $\mathbf{X}$ are Gaussian random variables, we have
$$
P(\|\mathbf{X}\|_{\infty}>\sqrt {t^2+2\log pn})\leq 2pn\exp\left\{-(t^2+2\log pn)/(2\sigma^2_{\mathbf{X}})\right\} 
\leq 2 \exp(-t^2/(2\sigma^2_{\mathbf{X}})).
$$
Together with Assumption \ref{ap:s0}, we have $\|\mathbf{X}\|_\infty=O_p(\sqrt{\log p} )$.
By Cauchy-Schwarz inequality, we have
\begin{equation}\label{eq:mclt3-0}
	\begin{aligned}
		n^{-1-\nu/2} \sum_{t=1}^{n}E|Z_{tj}|^{2+\nu} 
		&\leq\frac{\|Z_j\|_2^2}{n} \left(\frac{\|Z_j\|_{\infty}}{\sqrt{n}} \right)^{\nu}\\
		&\leq \frac{\|Z_j\|_2^2}{n} \left( \frac{(1+\|\gamma_j\|_1)\|\mathbf{X}\|_{\infty}}{\sqrt{n}} \right)^{\nu} \\
		&=O_p\left[\left(q_j\log p/n\right)^{\nu/2}\right].
	\end{aligned} 
\end{equation}
Then, we have, in probability,
\begin{equation}\label{eq:mclt3}
	\begin{aligned}
		v_n^{-2}n^{-1}\sum_{t=1}^{n}
		E(Z^2_{tj}e^2_{it}1_{\{|n^{-1/2}Z_{tj}e_{it}|\geq \rho v_n\}})
		&\leq v_n^{-2-\nu}n^{-1-\nu/2} \rho^{-\nu}
		\sum_{t=1}^{n}E|Z_{tj}e_{it}|^{2+\nu} \\
		&\leq v_n^{-2-\nu}n^{-1-\nu/2} \rho^{-\nu} L
		\sum_{t=1}^{n}E|Z_{tj}|^{2+\nu} \\
		&\rightarrow 0,
	\end{aligned} 
\end{equation}
where the first inequality is due to making use of $1_{\{|n^{-1/2}Z_{tj}e_{it}|\geq \rho v_n\}}$, the second holds because $e_{it}$ is normal, and the third is due to (\ref{eq:mclt3-0}) and Assumption \ref{ap:sj}.
Together with (\ref{eq:mclt1}), (\ref{eq:mclt2}) and (\ref{eq:mclt3}), the martingale central limit theorem (Theorem 5.3.4 in \cite{Fuller:1996wa}) implies
	\begin{equation}\label{eq:Znormal}
		\frac{1}{\sqrt{n}\sigma_i\tau_j}\sum_{t=1}^{n}Z_{tj}e_{it}\stackrel { d } { \rightarrow } \mathcal { N } ( 0,1 ).
	\end{equation}
	
Secondly, we have
\begin{equation}\label{eq:ZhatZ}
		\begin{aligned}
	\frac{1}{\sqrt{n}}\sum_{t=1}^{n}(\hat{Z}_{tj}-Z_{tj})e_{it}=&
	\frac{1}{\sqrt{n}}\sum_{t=1}^{n}\sum_{l\neq j}(\gamma_{jl}-\hat{\gamma}_{jl})X_{tl}e_{it} \\
	\leq& \sqrt{n} \|\gamma_j-\hat\gamma_j\|_1\|\mathbf{X}^{\T}\varepsilon_i/n\|_\infty \\
	=&O_p\left(q_j\log p/\sqrt n\right)
	\end{aligned}
\end{equation}
where the first equality is due to the definition, the second inequality is due to H{\"o}lder inequality, and the third equality is due to (\ref{eq:ndl-1norm}) and (\ref{eq:L2}). By (\ref{eq:Z-2-j}), we have
\begin{equation}\label{eq:tjZhat}
	\frac{\tau_j}{\| \hat{Z} _ { j }\|_2}\stackrel { P }\rightarrow 1.
\end{equation}
By Slutsky theorem, (\ref{eq:Znormal}), (\ref{eq:ZhatZ}), (\ref{eq:tjZhat}) and Assumption \ref{ap:sj}, we obtain (\ref{eq:asym-norm-final}).

\textit{Step} 3. By Theorem \ref{thm:Lasso}, we have $\hat s_{i}=O_p( s_{i} )=o_p(n)$. Then $(n - \hat s_{i})\hat { \sigma } _ { i} ^ { 2 }/n$ is asymptotically equivalent to $\hat { \sigma } _ {i } ^ { 2 }$. We have
\begin{equation}
\begin{aligned}
	|(n - \hat s_{i})\hat { \sigma } _ { i} ^ { 2 }/n-\sigma _ { i}^2|=&|(\hat\varepsilon_i^{\T}\hat\varepsilon_i-n{ \sigma }_{ i } ^ { 2 })/n| \\=& |(\hat\varepsilon_i-\varepsilon_i)^{\T}(\hat\varepsilon_i-\varepsilon_i)/n + 2\varepsilon_i^{\T}(\hat\varepsilon_i-\varepsilon_i)/n +(\varepsilon_i^{\T}\varepsilon_i-n{ \sigma }_{ i} ^ { 2 })/n| \\
		=&|\|\mathbf { X }(\hat a^{\rm Lasso}_{i} -a_i )\|_2^2/n-2\varepsilon_i^{\T}\mathbf { X }(\hat a^{\rm Lasso}_{i} -a_i )/n+(\varepsilon_i^{\T}\varepsilon_i-n{ \sigma }_{ i } ^ { 2 })/n|\\
		=& O_p(s_{i}\log (p)/n) +
			O_p(\sqrt{\log (p)/n})O_p(s_{i}\log (p)/n)
		+O_p\left( 1/\sqrt n \right)\\
		=&o_p(1),
\nonumber	
\end{aligned}	
\end{equation}
where the first three equalities are due to definitions and simple algebra, the fourth equality is due to Theorem \ref{thm:Lasso}, H{\"o}lder inequality, (\ref{eq:L2}) and the central limit theorem ($\varepsilon_i^{\T}\varepsilon_i-n{ \sigma }_{ i } ^ { 2 }=O_p(\sqrt n)$), the fifth equality is due to Assumption \ref{ap:s0}. 
By Theorem \ref{thm:Lasso} ($\hat s_{i}=O_p( s_{i} )=o_p(n)$), we obtain
	\begin{equation}
		\hat { \sigma } _ i ^ { 2 }-\sigma _ i^2=o_p(1).
		\label{eq:sigma-consis}
	\end{equation}
	
In all, with (\ref{eq:Delta-final}), (\ref{eq:asym-norm-final}) and (\ref{eq:sigma-consis}), the results of Theorem \ref{thm:dbl} follow.
\end{proof}

\subsection{Proofs of the theoretical results in Section \ref{sec:bdbl-th}}

We denote $\hat \varepsilon_{\rm cent}:=(\varepsilon_{{\rm cent},1},...,\varepsilon_{{\rm cent},n})$, where $\varepsilon_{{\rm cent},t}:=\hat \varepsilon_{it} - \hat\varepsilon_{i\cdot}$. Note that we omit the subscript $i$ for simplicity.
When we are conditional on $\mathbf{X}$ and $Y_i$ (therefore $\hat \varepsilon_{\rm cent}$), the only randomness comes from $\varepsilon^*_i$ which is generated from residual bootstrap or multiplier wild bootstrap. 
Since two bootstrap methods have a lot in common during the proof, we use the same symbol $\varepsilon^*_i$ to refer to the residuals generated by them, and discuss them separately when necessary.

\begin{proposition}
	(a) Under Assumption \ref{ap:s0_bt}, we have
	\begin{equation}
		\|\mathbf{X}\|_\infty=O_p(\sqrt{\log p} ). \label{eq:X_infty}
	\end{equation}
	(b) Under Assumption \ref{ap:s0_bt}, for tuning parameter satisfying $\lambda \asymp \sqrt{\log(p)/n}$, we have
	\begin{equation}
		E^*(\|\varepsilon^{*}_{i}\|^{2+\phi}_{2+\phi}/n)=O_{p}(1), \label{eq:e_hat_2+norm}
	\end{equation}
	where $0\leq\phi<1$. In particular, when $\phi=0$, we have
	\begin{equation}
		E^*(\|\varepsilon^{*}_{i}\|^{2}_{2}/n)=O_{p}(1), \label{eq:e_hat_2norm}
	\end{equation}
\end{proposition}

\begin{proof}
Let $\sigma^2_{\mathbf{X}}=\max_{1\leq t\leq n,1\leq j\leq p}{\rm Var}(X_{tj})$. By Propposition \ref{prop:prop2.3}, we have $\sigma^2_{\mathbf{X}}=O(1)$. Since the elements of $\mathbf{X}$ are Gaussian random variables, we have
$$
P(\|\mathbf{X}\|_{\infty}>\sqrt {t^2+2\log pn})\leq 2pn\exp\left\{-(t^2+2\log pn)/(2\sigma^2_{\mathbf{X}})\right\} 
\leq 2 \exp(-t^2/(2\sigma^2_{\mathbf{X}})).
$$
Together with Assumption \ref{ap:s0_bt}, (\ref{eq:X_infty}) follows.
	
Since $\varepsilon_i=(e_{i1},...,e_{in})^{\T}$ is Gaussian distributed, we have $E\|\varepsilon_i\|_{2+\phi}^{2+\phi}=O(n)$. We denote $\varepsilon_{i\cdot}:=\sum_{t=1}^{n}e_{it}/n $, then
$$
|\hat{\varepsilon}_{i\cdot}|\leq |\hat{\varepsilon}_{i\cdot}-{\varepsilon}_{i\cdot}|+|{\varepsilon}_{i\cdot}|\leq \|\hat{\varepsilon}_{i}-{\varepsilon}_{i}\|_1/n+|{\varepsilon}_{i\cdot}|=o_p(1)
$$
where the first two inequality are due to triangle inequality and the third is due to Theorem \ref{thm:Lasso}, Assumption \ref{ap:s0_bt} and the strong law of lager numbers. By Theorem \ref{thm:Lasso}, (\ref{eq:X_infty}) and Assumption \ref{ap:s0_bt}, we have
$$
	\left\|\hat{\varepsilon}_{\mathrm{cent}}-\varepsilon_i\right\|_{\infty} \leq |\hat{\varepsilon}_{i\cdot}|+\left\|\hat{\varepsilon}_{i}-\varepsilon_i\right\|_{\infty} \leq 
	|\hat{\varepsilon}_{i\cdot}|+\|\mathbf{X}\|_{\infty}\|\hat a^{\rm Lasso}_{i} -a_i\|_{1} = o_p(1).
$$
Thus, we have
$$
	\left\|\hat{\varepsilon}_{\mathrm{cent}}\right\|_{2+\phi}^{2+\phi} / n=O_p(1).
$$
For the residual bootstrap,
$$
	E^*(\|\varepsilon^{*}_{i}\|^{2+\phi}_{2+\phi}/n)=\left\|\hat{\varepsilon}_{\mathrm{cent}}\right\|_{2+\phi}^{2+\phi} / n=O_p(1).
$$
For the multiplier wild bootstrap,
$$
	E^*(\|\varepsilon^{*}_{i}\|^{2+\phi}_{2+\phi}/n)=(E|W_{11}|^{2+\phi})\left\|\hat{\varepsilon}_{\mathrm{cent}}\right\|_{2+\phi}^{2+\phi} / n=O_p(1).
$$
\end{proof}

\begin{proof}[\textbf{Proof of Theorem \ref{thm:BtLasso} (a)}]
Recall that we intend to prove
\begin{gather}
	\|\hat a^{\rm Lasso*}_{i}-\hat { a }^{\rm Lasso}_{i} \| _ { 1 } = O_p(s_{i}\log (p)/\sqrt{n}), \nonumber \\
	\|\mathbf{X}(\hat a^{\rm Lasso*}_{i}-\hat { a }^{\rm Lasso}_{i} )\|_2^2/n=O_p(s_{i}\log^2 (p)/n). \nonumber
\end{gather}
By Theorem \ref{thm:Lasso}, we have $\hat s_{i}=O_p(s_{i})$, thus, we only need to prove
\begin{gather}
	\|\hat a^{\rm Lasso*}_{i}-\hat { a }^{\rm Lasso}_{i} \| _ { 1 } = O_p(\hat s_{i}\log (p)/\sqrt{n}), \label{eq:BtLasso-1} \\
	\|\mathbf{X}(\hat a^{\rm Lasso*}_{i}-\hat { a }^{\rm Lasso}_{i} )\|_2^2/n=O_p(\hat s_{i}\log^2 (p)/n). \label{eq:BtLasso-2}
\end{gather}
Recall that
\begin{equation}
	\hat a^{\rm Lasso*}_{i} := \mathop{{\rm argmin}}\limits_{\alpha\in \mathbb{R}^p}\{ \|Y_i^* - \mathbf { X } \alpha\|_2^2/n + 2\lambda^* \|\alpha\|_1\},
	\nonumber
\end{equation}
we obtain the basic inequality
\begin{equation}
	 \|Y_i^* -\mathbf{X}\hat a^{\rm Lasso*}_{i}\|_2^2/n + 2\lambda^*\|\hat a^{\rm Lasso*}_{i}\|_1\leq
	 \|Y_i^* -\mathbf{X}\hat a^{\rm Lasso}_{i}\|_2^2/n + 2\lambda^*\|\hat a^{\rm Lasso}_{i}\|_1.
	 \nonumber
\end{equation}
We denote $\delta^*:= \hat a^{\rm Lasso*}_{i}-\hat { a }^{\rm Lasso}_{i}$. By simple algebra, we have 
\begin{equation}
	\| \mathbf { X } \delta^* \| _ { 2 } ^ { 2 } / n \leq
	2\varepsilon^{*\T}_i\mathbf { X } \delta^*/n + 2\lambda^*(\|\hat a^{\rm Lasso}_{i}\|_1-\|\hat a^{\rm Lasso*}_{i}\|_1).
	\label{eq:ndl-basic-bt}
\end{equation}
We can obtain a bound for $\|\mathbf{X}^{\T}\varepsilon^*_i/n\|_\infty$ as follows:
\begin{equation}
	\begin{aligned}
		P^*\left(\|\mathbf{X}^{\T}\varepsilon^*_i/n\|_\infty> 
		C_2\frac{\log p} {\sqrt n}\right)
		&\leq \frac{E^*(\|\mathbf{X}^{\T}\varepsilon^*_i\|_\infty)^2}{C_2 n\log^2 p} \leq \frac{8\log(2p)\|\mathbf{X}\|^2_\infty E^*(\|\varepsilon^{*}_i\|^2_2)}{C_2n\log^2 p},	
	\end{aligned}
	\nonumber
\end{equation}
where the first inequality is due to Markov inequality, the second inequality is due to Nemirovski's inequality and equation (6.5) in \cite{Buhlmann:2011}. By (\ref{eq:X_infty}) and (\ref{eq:e_hat_2norm}), we have, in probability,
\begin{equation}
	\|\mathbf{X}^{\T}\varepsilon^*_i/n\|_\infty\leq
		C_2\frac{\log p} {\sqrt n}.
	\label{eq:Xe-infty-bt}
\end{equation}
With suitable chosen $\lambda^*\asymp\log p/\sqrt{n}$, we have, in probability,
\begin{equation}
	\|\mathbf{X}^{\T}\varepsilon^*_i/n\|_\infty\leq
		C_2\frac{\log p} {\sqrt n}\leq\lambda^*/2.
	\nonumber
\end{equation}
Furthermore, by H\"older inequality, we have
\begin{equation}
	\varepsilon^{*\T}_i\mathbf { X } \delta^*/n\leq
	 \|\mathbf{X}^{\T}\varepsilon^*_i/n\|_\infty\| \delta^*\|_1\leq\lambda^*\| \delta^*\|_1/2.
	\nonumber
\end{equation}
Recall that $\hat S_i=\{j\in \{1,...,p\}:\hat a^{\rm Lasso}_{ij}\ne0\}$, by triangle inequality and $\hat a^{\rm Lasso}_{i\hat S_i^C}=0$, we have
\begin{equation}
	\|\hat a^{\rm Lasso}_{i}\|_1-\|\hat a^{\rm Lasso*}_{i}\|_1
	=\|\hat a^{\rm Lasso}_{i\hat S_i}\|_1-\|\hat a^{\rm Lasso*}_{i\hat S_i}\|_1 +
	\|\hat a^{\rm Lasso}_{i\hat S_i^C}\|_1-\|\hat a^{\rm Lasso*}_{i\hat S_i^C}\|_1\leq 
	\|\delta^*_{\hat S_i}\|_1-\|\delta_{\hat S_i^C}^*\|_1.
	\nonumber
\end{equation}
Therefore, (\ref{eq:ndl-basic-bt}) becomes
\begin{equation}
	\begin{aligned}
		0\leq\| \mathbf { X } \delta^* \| _ { 2 } ^ { 2 } / n
		&\leq \lambda^*\| \delta^*\|_1+ 2\lambda^*(\|\delta^*_{\hat S_i}\|_1-\|\delta_{\hat S_i^C}^*\|_1)\\
		&= \lambda^*(\|\delta^*_{\hat S_i}\|_1+\|\delta_{\hat S_i^C}^*\|_1)+ 2\lambda^*(\|\delta^*_{\hat S_i}\|_1-\|\delta_{\hat S_i^C}^*\|_1)\\
		&=3\lambda^*\|\delta^*_{\hat S_i}\|_1-\lambda^*\|\delta_{\hat S_i^C}^*\|_1\\
		&\leq 3\lambda^*\|\delta^*\|_1.
	\end{aligned}
	\label{eq:btl-1}
\end{equation}
Inequality (\ref{eq:btl-1}) implies $\|\delta^*_{\hat S_i^C}\|_1\leq3\|\delta^*_{\hat S_i}\|_1$ so that 
\begin{equation}
	\|\delta^*\|_1\leq 4\|\delta^*_{\hat S_i}\|_1\leq4\sqrt{\hat s_{i}}\|\delta^*\|_2.
	\label{eq:btl-2}
\end{equation}
By (\ref{eq:ndl-RE}) and (\ref{eq:ndl-2}), we have
\begin{equation}
	\| \mathbf { X } \delta^* \| _ { 2 } ^ { 2 } / n\geq \alpha\|\delta \|_2^2-\tau \|\delta \|_1^2\geq
	\left(\frac{\alpha}{4\hat s_{i}}-\tau \right)\|\delta\|_1^2.
	\label{eq:btl-3}
\end{equation}
By Theorem \ref{thm:Lasso} and Proposition \ref{prop:4.2&4.3}, we have $\tau \hat s_{i}=O_p(\tau s_{i})=O_p(1)$. 
Combining with (\ref{eq:btl-1}), (\ref{eq:btl-3}) and Assumption \ref{ap:lambda_bt} ($\lambda^* \asymp \log(p)/\sqrt{n}$), we obtain (\ref{eq:BtLasso-1}). Then by (\ref{eq:btl-1}), (\ref{eq:BtLasso-1}) and Assumption \ref{ap:lambda_bt}, we obtain (\ref{eq:BtLasso-2}).
\end{proof}

\begin{proof}[\textbf{Proof of Theorem \ref{thm:BtLasso} (b)}]
	By Theorem \ref{thm:Lasso}, we have $\hat s_{i}=O_p(s_{i})$. To prove $\hat s^*_{i}=O_p(s_{i})$, we only need to prove $\hat s^*_{i}=O_p(\hat s_{i})$.
	
	The proof is the same as the counterpart in Theorem \ref{thm:Lasso} (b) except for the bound (\ref{eq:Xe_inf}). We only need to replace
	\begin{equation}
		\|\mathbf{X}^{\T} \varepsilon_i/n\|_{\infty}\leq C_1 \sqrt{\log p/n},
		\nonumber
	\end{equation}
	by
	\begin{equation}
		\|\mathbf{X}^{\T} \varepsilon^*_i/n\|_{\infty}\leq 	C_2\log p/\sqrt n,
		\nonumber
	\end{equation}
	where the second inequality is due to (\ref{eq:Xe-infty-bt}). With
	$$
	\lambda^* \geq \frac{4C_2c^*}{c_*}\frac{\log p}{\sqrt{n}},
	$$
	the result follows from the same reasoning.
\end{proof}

Similar to (\ref{eq:KKT-3recall}), we have the bootstrap analogue,
\begin{equation}
	\hat{a}_{i}^*-\hat { a }^{\rm Lasso}_{i}    = \hat\Theta\mathbf { X } ^ {\rm  T } \varepsilon^*_i / n + (I-\hat\Theta\hat { \Sigma })(\hat { a }^{\rm Lasso*}_{i} -\hat { a }^{\rm Lasso}_{i}) .
	\nonumber
\end{equation}
Since we use the same $\mathbf{X}$ for nodewise Lasso during the bootstrap procedure, the bounds proved in the proof of Theorem \ref{thm:dbl} still hold.

\begin{proof}[\textbf{Proof of Theorem \ref{thm:bdbl}}]
We prove the theorem in three steps. Step 1 proves that $(I-\hat\Theta\hat { \Sigma })(\hat { a }^{\rm Lasso*}_{i} -\hat { a }^{\rm Lasso}_{i})$ is asymptotically negligible. Step 2 proves that $\hat\Theta\mathbf { X } ^ {\rm  T } \varepsilon^*_i / n$ is asymptotically normal. Step 3 proves that our variance estimator is consistent.

\textit{Step} 1. By H{\"o}lder inequality, we have
\begin{equation}
	\begin{aligned} 
	\sqrt { n } \| (I-\hat\Theta\hat { \Sigma })(\hat { a }^{\rm Lasso*}_{i} -\hat { a }^{\rm Lasso}_{i}) \| _ { \infty } 
	& \leq \sqrt { n }\| (I- \hat { \Theta } \hat { \Sigma }  ) \| _ { \infty } \| \hat { a }^{\rm Lasso*}_{i} -\hat { a }^{\rm Lasso}_{i}\| _ { 1 } \\
	& \leq \sqrt { n } \left( \max _ { j } \lambda _ { j } / \hat { \tau } _ { j } ^ { 2 } \right) \| \hat { a }^{\rm Lasso*}_{i} -\hat { a }^{\rm Lasso}_{i}\| _ { 1 } .
	\end{aligned}
	\label{eq:Delta-1-bt}
\end{equation}
By Theorem \ref{thm:BtLasso}, we have 
\begin{equation}\label{eq:Delta-2-bt}
	\| \hat { a }^{\rm Lasso*}_{i} -\hat { a }^{\rm Lasso}_{i} \| _ { 1 } = O_p (s_{i}\log p/\sqrt{n}).
\end{equation}
Since we use the same $\mathbf{X}$ for nodewise Lasso during the bootstrap procedure, (\ref{eq:tau-const}) still holds.
Together with Assumption \ref{ap:lambda_bt} and \ref{ap:s0_bt}, we have
\begin{equation}
	\sqrt { n } \| (I-\hat\Theta\hat { \Sigma })(\hat { a }^{\rm Lasso*}_{i} -\hat { a }^{\rm Lasso}_{i}) \| _ { \infty } =o _ p(1).
	\label{eq:Delta-final-bt}
\end{equation}

\textit{Step} 2. Compared to the counterpart in Theorem \ref{thm:dbl}, $\hat{Z}_{ij}$ is no longer correlated with $\varepsilon^*_i$ when we conditional on $\mathbf{X}$ and consider bootstrap measure.
By (\ref{eq:e_hat_2+norm}), 
$$
\left\{\frac{1}{\|\hat\varepsilon_{\rm cent}\|_2^2/n\cdot\| \hat{Z} _ { j }\|_2 }\hat{Z}_{tj}\varepsilon_{1t}^*,t=1,...,n\right\}
$$ 
are independent variables meet the Lyapunov condition. By central limit theorem, 
	\begin{equation}
		\frac{1}{\|\hat\varepsilon_{\rm cent}\|_2^2/n\cdot\| \hat{Z} _ { j }\|_2 }\sum_{t=1}^{n}\hat{Z}_{tj}\varepsilon_{1t}^* \stackrel { d^* } { \rightarrow } \mathcal { N } ( 0,1 )\; \text{in probability}.
		\label{eq:asym-norm-final-bt}
	\end{equation}

\textit{Step} 3. By Theorem \ref{thm:BtLasso}, we have $\hat s^*_{i}=O_p( s_{i} )=o_p(n)$.
Then $(n - \hat s^*_{i})\hat { \sigma } _ { i} ^ { *2 }/n$ is asymptotically equivalent to $\hat { \sigma } _ {i } ^ { *2 }$.
Also, we have
$$
\begin{aligned}
		&|(n - \hat s^*_{i})\hat { \sigma } _ { i} ^ { *2 }/n-\|\hat\varepsilon_{\rm cent}\|_2^2/n|\\
			=&\big|(\hat\varepsilon^{*\T}_i\hat\varepsilon^*_i-\|\hat\varepsilon_{\rm cent}\|_2^2)/n\big|\\
			=& \big|(\hat\varepsilon^*_i-\varepsilon^*_i)^{\T}(\hat\varepsilon^*_i-\varepsilon^*_i)/n + 2\varepsilon^{*\T}_i(\hat\varepsilon^*_i-\varepsilon^*_i)/n +(\varepsilon^\T_i\varepsilon^*_i-\|\hat\varepsilon_{\rm cent}\|_2^2)/n\big| \\
			=&\big|\|\mathbf{X}(\hat a^{\rm Lasso*}_{i}-\hat { a }^{\rm Lasso}_{i} )\|_2^2/n-2\varepsilon_i^{*\T}\mathbf { X }(\hat a^{\rm Lasso*}_{i}-\hat { a }^{\rm Lasso}_{i} )/n +(\varepsilon_i^{* \T}\varepsilon^*_i-\|\hat\varepsilon_{\rm cent}\|_2^2)/n\big| \\
			= & O_p(s_{i}\log^2 (p)/n) + 
				O_{p}\left( \log(p)/\sqrt n \right)O_{p}\left(s_{i} \log(p)/\sqrt n \right)+O_{p}\left(1/{\sqrt n} \right)\\
				=&o_p(1),
\end{aligned}	
$$
where the first three equalities are due to definitions and simple algebra, the fourth equality is due to Theorem \ref{thm:BtLasso}, (\ref{eq:Xe-infty-bt}), central limit theorem, and (\ref{eq:e_hat_2norm}), the fifth equality is due to Assumption \ref{ap:s0_bt}.
By Theorem \ref{thm:BtLasso} ($\hat s^*_{i}=O_p( s_{i} )=o_p(n)$), we obtain,
	\begin{equation}
		  \hat { \sigma } _ { i } ^ {* 2 }-\|\hat\varepsilon_{\rm cent}\|_2^2/n=o_{p}(1)\; \text{in probability}.
	\label{eq:sigma-consis-bt}
	\end{equation}
	
In all, with (\ref{eq:Delta-final-bt}), (\ref{eq:asym-norm-final-bt}) and (\ref{eq:sigma-consis-bt}), the results of Theorem \ref{thm:bdbl} follow.
\end{proof}

\section{Additional Simulation Results}

Figures \ref{fig:cp2} and \ref{fig:length2} show the coverage probabilities and average confidence interval lengths for homoscedastic non-Gaussian errors, heteroscedastic Gaussian errors and heteroscedastic non-Gaussian errors. Compared to homoscedastic Gaussian errors, different distributions of errors do not lead to significant difference of performance. Again, we can see that the LDPE has honest coverage probabilities and the BtLDPE and MultiBtLDPE have shorter confidence interval lengths compared to the LDPE.

\clearpage

\begin{figure}[t]
\centering
\includegraphics[scale=0.8]{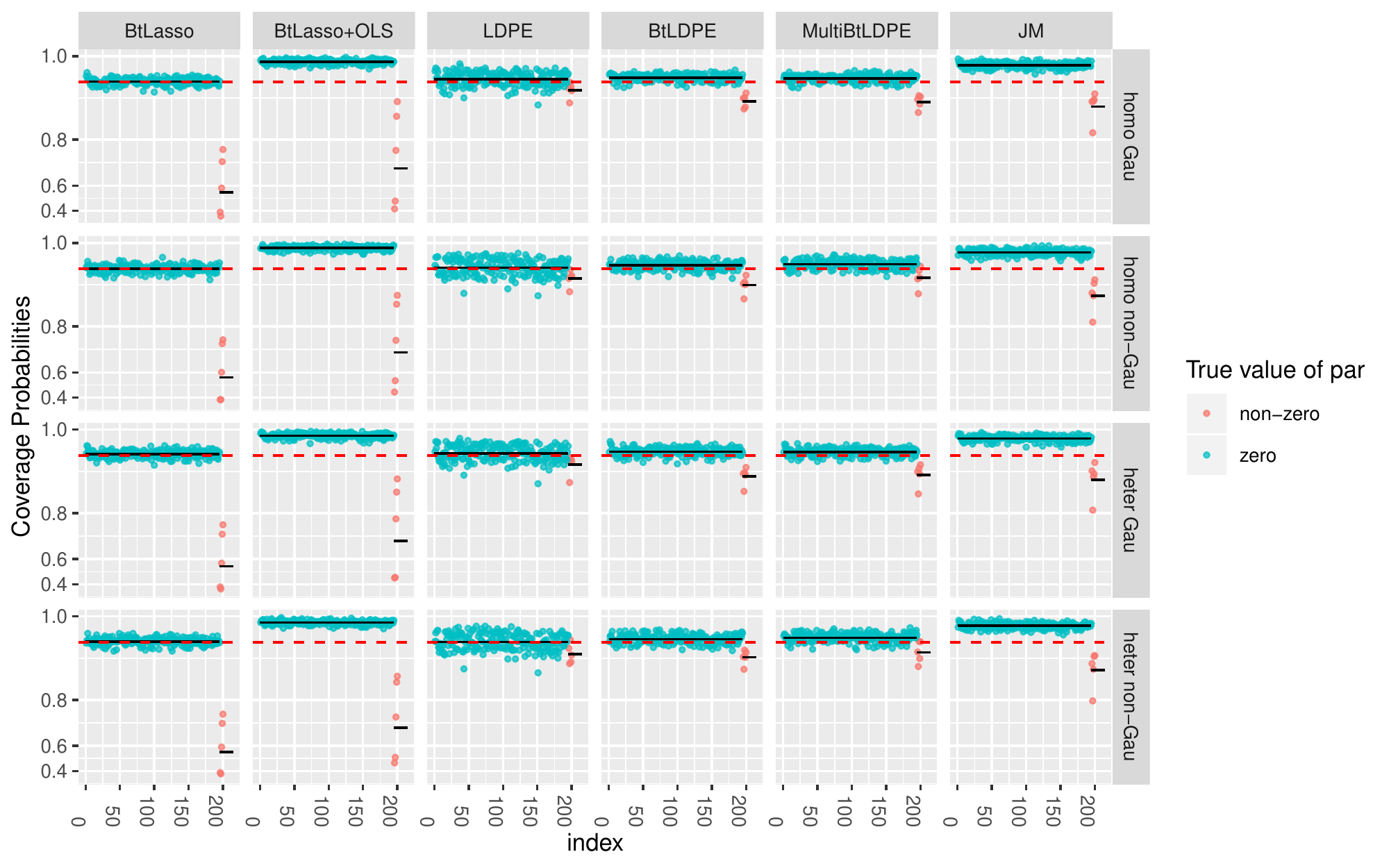}\par
\caption{
Comparison of empirical coverage probabilities for 1000 replications produced by six methods (columns) in four cases (rows). We set $n=100$ and $s_{i}=5$. Different rows correspond to different distributions of errors. Index on the $x$-axis corresponds to different $a_{1j}$'s, which are arranged from small to large in absolute values. The first $p-5$ elements of $a_{1j}$'s are zeros (blue points) and the last 5 are non-zeros (red points). The black lines are the total averages of coverage probabilities for zero and non-zero parameters respectively. The red dashed lines correspond to the nominal confidence level 95\%.
}
\label{fig:cp2}
\end{figure}

\begin{figure}[t]
\centering
\includegraphics[scale=0.8]{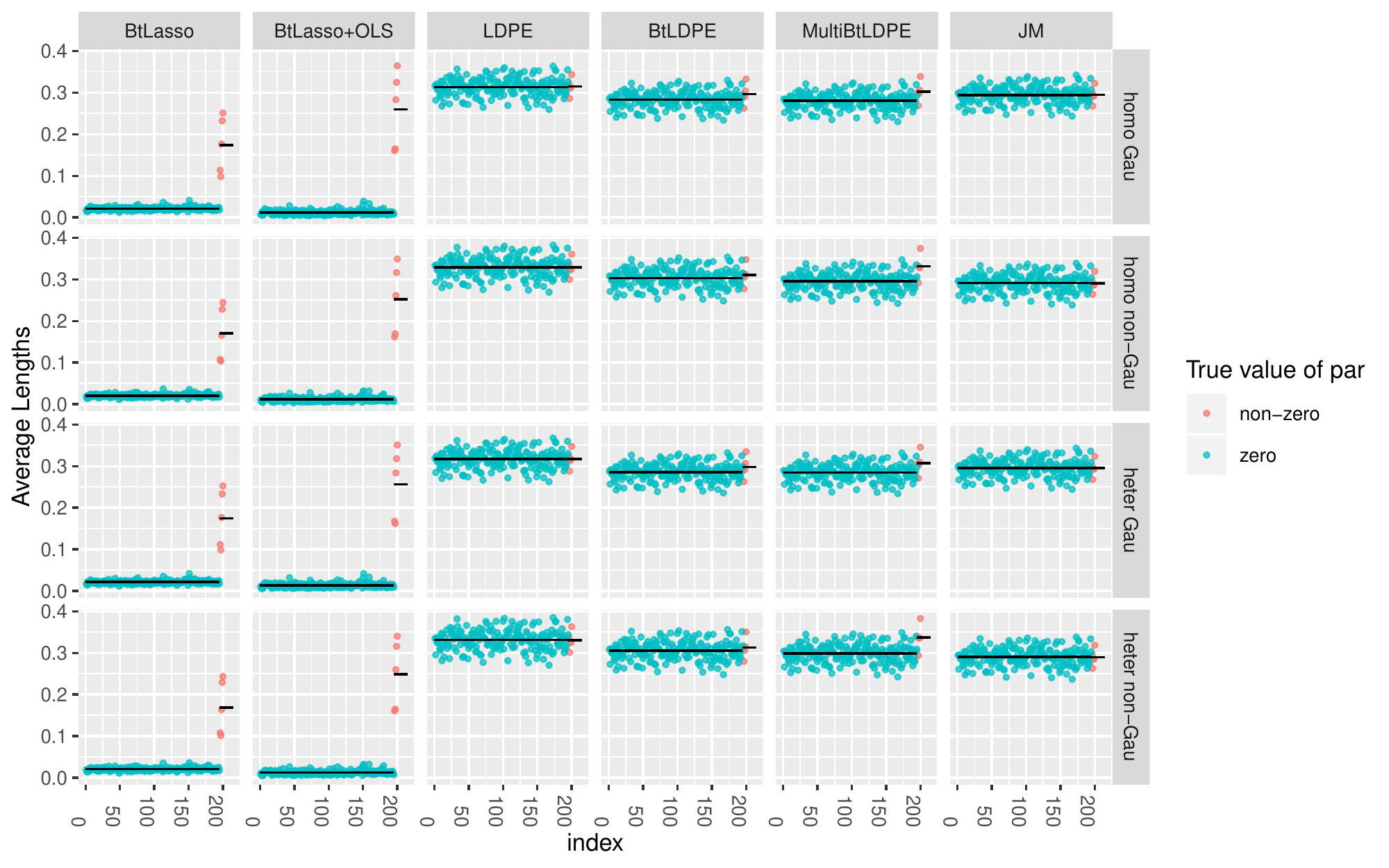}\par
\caption{
Comparison of average confidence interval lengths for 1000 replications produced by six methods (columns) in four cases (rows). We set $n=100$ and $s_{i}=5$. Different rows correspond to different distributions of errors. Index on the $x$-axis corresponds to different $a_{1j}$'s, which are arranged from small to large in absolute values. The first $p-5$ elements of $a_{1j}$'s are zeros (blue points) and the last 5 are non-zeros (red points). The black lines are the total averages of confidence interval lengths for zero and non-zero parameters respectively.
}
\label{fig:length2}
\end{figure}

\end{document}